\documentclass[lettersize,journal]{IEEEtran}
\usepackage{amsmath,amsfonts,amssymb}
\usepackage[ruled,vlined]{algorithm2e}
\usepackage{array}
\usepackage[caption=false]{subfig}
\usepackage{textcomp}
\usepackage{subfloat}
\usepackage{url}
\usepackage{verbatim}
\usepackage{graphicx}
\usepackage{cite}
\usepackage[n,  
    advantage,
    operators,
    sets,
    adversary,
    landau,
    probability,
    notions,
    logic,
    ff, 
    mm,
    primitives,
    events,
    complexity,
    oracles,
    asymptotics,
    keys
]{cryptocode}
\renewcommand{\advantage}[2][n]{\textnormal{\textsf{Adv}}_{\mathrm{#1}}(#2)}
\renewcommand{\prob}[1]{\textnormal{\textsf{Pr}}[#1]}
\usepackage{xcolor}
\definecolor{myblue}{HTML}{0000CC}
\usepackage[colorlinks=true]{hyperref}
\hypersetup{
    linkcolor=myblue,
    urlcolor=myblue,
    citecolor=myblue
}
\usepackage{booktabs} 
\usepackage{multirow}
\usepackage{enumerate}
\usepackage{amsthm}
\usepackage{stmaryrd}
\usepackage{threeparttable}
\usepackage{bm}
\usepackage{ulem}
\usepackage{bbding}

\newtheorem{theorem}{\bf Theorem}
\newtheorem{myDef}{Definition}

\newcommand{\lowercasenote}[1]{\textnormal{\MakeLowercase{#1}}}

\usepackage{algpseudocode}
\usepackage{stmaryrd} 
\usepackage{mdframed}

\definecolor{cpcolor}{RGB}{245, 245, 245} 
\definecolor{cspcolor}{RGB}{215, 215, 215} 
\definecolor{whilecolor}{RGB}{253, 253, 253} 

\mdfdefinestyle{leftline}{%
  topline=false,%
  rightline=false,%
  bottomline=false,%
  innertopmargin=0pt,%
  innerrightmargin=0pt,%
  innerbottommargin=0pt,%
  leftmargin=6pt,%
  linewidth=1pt,%
  linecolor=black,%
}
\mdfdefinestyle{inner1}{
topline=true,rightline=true,bottomline=true,linewidth=0.5pt,innertopmargin=4pt,backgroundcolor=cpcolor,innerleftmargin=2pt,innerrightmargin=2pt,innerbottommargin=2pt, leftmargin=0pt
}
\mdfdefinestyle{inner2}{
backgroundcolor=whilecolor,innertopmargin=4pt,innerbottommargin=4pt, innerbottommargin=4pt, leftmargin=0pt
}

\def\BibTeX{{\rm B\kern-.05em{\sc i\kern-.025em b}\kern-.08em
    T\kern-.1667em\lower.7ex\hbox{E}\kern-.125emX}}
\usepackage{balance}

\begin{document}
\title{MORSE: An Efficient Homomorphic Secret Sharing Scheme Enabling Non-Linear Operation}
\author{Weiquan Deng, Bowen Zhao, Yang Xiao, Yantao Zhong, Qingqi Pei, Ximeng Liu
\thanks{Manuscript received x x, x; revised x x, x.}}


\maketitle

\begin{abstract}
Homomorphic secret sharing (HSS) enables two servers to locally perform functions on encrypted data directly and obtain the results in the form of shares.
A Paillier-based HSS solution seamlessly achieves multiplicative homomorphism and consumes less communication costs.
Unfortunately, existing Paillier-based HSS schemes suffer from a large private key size, potential calculation error, expensive computation and storage overhead, and only valid on linear operations (e.g., addition and multiplication).
To this end, inspired by the Paillier cryptosystem with fast encryption and decryption, we propose MORSE, an efficient homomorphic secret sharing scheme enabling non-linear operation, which enjoys a small key size, no calculation error and low overhead.
In terms of functions, MORSE supports addition, subtraction, multiplication, scalar-multiplication, and comparison.
Particularly, we carefully design two conversion protocols achieving the mutual conversion between one Paillier ciphertext and two secret shares, which allows MORSE to continuously perform the above operations.
Rigorous analyses demonstrate that MORSE securely outputs correct results.
Experimental results show that MORSE makes a runtime improvement of up to 9.3 times in terms of secure multiplication, and a communication costs reduction of up to $16.6\%$ in secure comparison, compared to the state-of-the-art.

\end{abstract}


\section{Introduction}
\IEEEPARstart{O}{utsourcing} data and computation to a cloud server is a flexible way to store data and perform operations on data.
Unfortunately, a cloud server being responsible for storage and computation on outsourced data is not always trusted, since it may intentionally or unintentionally leak data, leading to privacy issues. 
For example, a data breach at Twitter exposes the private information of over 5.4 million users \cite{Twitter_breaches}, and Pegasus Airline leaks 6.5 terabytes information due to a misconfiguration on AWS \cite{Pegasu_breaches}.
To avoid privacy leakage, users typically encrypt their data before outsourcing \cite{li2018privacy}.
Particularly, users also employ homomorphic encryption (HE) \cite{zhang2016review} including partially homomorphic encryption (PHE) and fully homomorphic encryption (FHE) to balance data usability and privacy.

HE enables operations over encrypted data (or say ciphertexts) directly, however, it suffers from several limitations. 
On the one hand, FHE suffers from expensive computation and storage overhead.
For instance, the operations of FHE on ciphertexts introduce noise, and when the noise exceeds a certain threshold, it results in decryption failure.
One powerful tool to mitigate the noise on ciphertexts is the bootstrapping operation introduced by Gentry \cite{gentry2009fully}.
Unfortunately, despite extensive researches, the bootstrapping operation still incurs high computation costs \cite{han2020better}.
Furthermore, CKKS \cite{cheon2017homomorphic}, a famous FHE implemented in Microsoft SEAL \cite{Microsoft_SEAL}, incurs 10 times to $10^4$ times storage overhead compared to plaintext \cite{akavia2023csher}. 
On the other hand, PHE support limited operations on ciphertexts and necessitates staggering communication costs to achieve more types of operation.
For example, the famous PHE, Paillier \cite{paillier1999public} algorithm, supports only homomorphic addition and lacks support for homomorphic multiplication.
To enable multiplicative homomorphism on Paillier ciphertexts and more types of operation, a twin-server architecture and the interactions between the two servers are widely adopted in the literature \cite{zhao2022soci, zhao2024soci+}.
Whereas, this method incurs substantial communication costs, resulting in performance degeneration in a bad network.

To mitigate the limitations of HE, HSS integrating PHE and secret sharing is introduced by scheme \cite{boyle2016breaking}. 
HSS works on the twin-server architecture, but requires no interaction between the two servers. 
Furthermore, HSS requires no Beaver multiplication triple \cite{beaver1997commodity} that is necessary for a secret sharing-based scheme to achieve multiplication. 
In terms of function, HSS supports continuous addition and multiplication, which removes the time-consuming bootstrapping operation required by FHE. 
Similar to secret sharing-based solutions, HSS stores data on two servers and allows the two servers to perform computation independently, while the computation result comes from the combination of that of the two servers.

\begin{table*}[!ht]
\centering
\label{Comparison of Different Paillier-Based HSS Scheme}
\caption{Functional comparison between existing schemes and MORSE}
\begin{threeparttable}
\begin{tabular}{@{}cccccccc@{}}
\toprule
Scheme 
& \begin{tabular}[c]{@{}c@{}}Small key size\end{tabular} 
& \begin{tabular}[c]{@{}c@{}}Negligible computation error \\ or no computation error  \end{tabular} 
& \begin{tabular}[c]{@{}c@{}}Low storage \\ overhead\end{tabular} 
& \begin{tabular}[c]{@{}c@{}}Efficient\end{tabular} 
& \begin{tabular}[c]{@{}c@{}}Non-linear operation\end{tabular} 
& \begin{tabular}[c]{@{}c@{}}Homomorphic \\ addition\end{tabular}
& \begin{tabular}[c]{@{}c@{}}Homomorphic \\ multiplication\end{tabular}
\\ 
\midrule
FGJS17 \cite{fazio2017homomorphic}        
& \XSolidBrush  & \XSolidBrush  & \XSolidBrush  & \XSolidBrush  & \XSolidBrush  & \Checkmark  & \Checkmark\\
OSY21 \cite{orlandi2021rise}       
& \XSolidBrush  & \Checkmark    & \XSolidBrush  & \XSolidBrush  & \XSolidBrush  & \Checkmark  & \Checkmark\\
RS21 \cite{roy2021large}       
& \XSolidBrush  & \Checkmark    & \XSolidBrush  & \XSolidBrush  & \XSolidBrush  & \Checkmark  & \Checkmark\\
MORSE   
& \Checkmark    & \Checkmark    & \Checkmark    & \Checkmark    & \Checkmark    & \Checkmark  & \Checkmark\\ 
\bottomrule
\end{tabular}
\end{threeparttable}

\end{table*}

Although HSS overcomes the limitations of HE, it faces with new challenges. 
Firstly, the solutions \cite{fazio2017homomorphic}, \cite{orlandi2021rise} and \cite{roy2021large} suffer from a large private key size, imposing a significant burden on the data owner (or client).
Secondly, solution \cite{fazio2017homomorphic} incurs potential calculation error. 
HSS involves a procedure solving discrete logarithm by the two servers, which is called distributed discrete logarithm (DDLog), share conversion, or distance function.
In particular, solution \cite{fazio2017homomorphic} features a non-negligible computation error in this procedure.
Although the error probability can be chosen as a small parameter, it consumes more computation costs.
Thirdly, solutions \cite{fazio2017homomorphic}, \cite{orlandi2021rise} and \cite{roy2021large} incur expensive computation and storage overhead. 
In computation overhead, solutions \cite{fazio2017homomorphic}, \cite{orlandi2021rise} and \cite{roy2021large} necessitate a huge amount of modular exponentiation operations.
In storage overhead, solutions \cite{fazio2017homomorphic,orlandi2021rise} store multiple ciphertexts involving the segments of a private key to achieve continuous computations, and the ciphertext space of solution 
\cite{roy2021large} is much larger than that of standard Paillier.
Lastly, existing solutions are only valid on linear operations (e.g., addition, and multiplication).
In particular, solutions \cite{fazio2017homomorphic}, \cite{orlandi2021rise} and \cite{roy2021large} lack support for non-linear operation (e.g., comparison). 

To tackle the aforementioned challenges, in this work, 
we propose MORSE\footnote{MORSE: an efficient ho\underline{MO}morphic sec\underline{R}et \underline{S}haring scheme enabling non-lin\underline{E}ar operation}, an efficient homomorphic secret sharing enabling non-linear operation.
Specifically, to improve computation efficiency and reduce storage overhead, we construct MORSE based on a variant of Paillier with fast encryption and decryption \cite{ma2021optimized}, which enjoys a small private key size.
Additionally, we adopt a perfectly correct DDLog \cite{orlandi2021rise} to avoid potential calculation error.
Particularly, MORSE supports computation over encrypted integers, which benefits from the feature of the Paillier cryptosystem \cite{ma2021optimized} and our careful design. 
Besides, a comparison protocol is proposed to make MORSE enable non-linear operation.
In short, the contributions of this work are three-fold.
\begin{itemize}
    \item \textbf{A novel homomorphic secret sharing}.
    Based on a variant of Paillier with fast encryption and decryption \cite{ma2021optimized}, we propose MORSE, an efficient homomorphic secret sharing designed to remedy the limitations of existing Paillier-based HSS schemes.
    MORSE features a small key size and effective operations on integer, and is capable of efficiently performing linear and non-linear operations. 
    \item \textbf{Continuous  operations support}.
    To achieve continuous operations on encrypted data, we carefully design a protocol for converting secret shares into Paillier ciphertexts (named \texttt{S2C}) and a protocol for converting Paillier ciphertexts into secret shares (named \texttt{C2S}).
    \item 
    \textbf{Linear and non-linear operations on integer}.
    MORSE enables not only linear operations (e.g., multiplication) but also non-linear operations (e.g., comparison). 
    Furthermore, MORSE supports secure operations over integer rather than natural number.  
    Particularly, MORSE outperforms the state-of-the-art in terms of runtime and communication costs.
\end{itemize}

The rest of this work is organized as follows.
Section \ref{Section_2} reviews the related work. 
Section \ref{Section_3} outlines the preliminaries necessary for constructing MORSE.
Section \ref{Section_4} describes the system model and the threat model.
In Section \ref{Section_5}, we detail the design of MORSE.
The correctness and security analyses are presented in Section \ref{Section_6}, while experimental evaluations are conducted in Section \ref{Section_7}.
Finaly, we conclude this work in Section \ref{Section_8}.

\section{Related Work}\label{Section_2}
Extensive researches focus on enabling Paillier \cite{paillier1999public} algorithm to perform secure multiplication on ciphertexts.
To facilitate multiplication on Paillier \cite{paillier1999public} ciphertexts, Elmehdwi \textit{et al.} \cite{elmehdwi2014secure} employed the twin-server architecture, capitalizing on the interactions between the two servers to achieve secure multiplication. 
Utilizing random numbers to mask the real values of private data during the interactions, the privacy of data can be guaranteed. 
However, the work \cite{elmehdwi2014secure} granted a server possession of a private key, potentially creating a single point of security failure.
To prevent the single point of security failure, the schemes \cite{liu2016privacy} and \cite{liu2016efficient} partitioned the private key of Paillier algorithm into distinct segments. 
Specifically, the work \cite{liu2016privacy} divided the private key into two parts, held by two non-colluding servers who together possess the capability to decrypt a ciphertext in a collaborative manner.
The work \cite{liu2016efficient} split the private key into multiple keys held by multiple servers, who can collaboratively decrypt a ciphertext.
Following the twin-server architecture of the work \cite{liu2016privacy}, Zhao \textit{et al.} \cite{zhao2022soci} implemented secure multiplication on Paillier ciphertexts by employing smaller random numbers to conceal the real values of privacy data, thereby achieving a more efficient secure multiplication compared to the work \cite{liu2016privacy}.
To further enhance the efficiency of the secure multiplication demonstrated in the work \cite{zhao2022soci}, the work \cite{zhao2024soci+} presented a novel $(2,2)$-threshold Paillier cryptosystem and an offline and online mechanism.
Utilizing these building blocks, the work \cite{zhao2024soci+} presented a secure multiplication protocol faster than its predecessor in the work \cite{zhao2022soci}.
Despite significant development, all the aforementioned solutions necessitate considerable interactions between the two servers, leading to a substantial communication burden.

The emergence of HSS addresses the issue of excessive communication costs for secure multiplication on Paillier ciphertexts.
The first HSS scheme is constructed by Boyle \textit{et al.} \cite{boyle2016breaking}, which is based on the ElGamal encryption scheme with additive homomorphism.
Inspired by Boyle \textit{et al.} \cite{boyle2016breaking}, the work FGJS17 \cite{fazio2017homomorphic} built a Paillier-based HSS.
However, the DDLog function in FGJS17 \cite{fazio2017homomorphic} necessitates iterative computations to resolve the distributed discrete logarithm, and FGJS17 \cite{fazio2017homomorphic} claimed that its DDLog has a certain probability of error in computation.
To this end, OSY21 \cite{orlandi2021rise} introduced a Paillier-based HSS, and presented a perfectly correct and efficient DDLog.
Unfortunately, to achieve continuous computations of secure multiplication,  FGJS17 \cite{fazio2017homomorphic} and OSY21 \cite{orlandi2021rise} partitioned the private key $sk$ into multiple segments (e.g., $sk_1$, $sk_2$, ..., $sk_l$) and generated the ciphertexts such as $\llbracket sk_i \cdot x \rrbracket$, resulting in significant computation and storage overhead. 
To avoid private key partitioning and the creation of multiple ciphertexts, RS21 \cite{roy2021large} designed an HSS scheme based on the Damg{\aa}rd–Jurik public-key encryption scheme \cite{damgaard2001generalisation}.
However, the ciphertext space of RS21 \cite{roy2021large} is larger than that of standard Paillier cryptosystem, which also incurs substantial storage overhead.
Moreover, the above Paillier-Based HSS schemes suffer from a huge private key size and non-support for non-linear operation,
and we make a comparison between them and MORSE in Table \ref{Comparison of Different Paillier-Based HSS Scheme}.

\section{Preliminaries} \label{Section_3}
\subsection{Paillier Cryptosystem with Fast Encryption and Decryption}
The work \cite{ma2021optimized} presents a Paillier cryptosystem characterized by fast encryption and decryption, which is called FastPai within this work, and it consists of the following components.

\subsubsection{\textbf{N Generation} \texttt{(NGen)}}
\texttt{NGen} takes a security parameter $\kappa$ (e.g., $\kappa=128$) as input and outputs a tuple $(N=P\cdot Q,P,Q,p,q)$, such that $P$, $Q$ are primes with $\frac{n(\kappa)}{2}$-bit and $p$, $q$ are odd primes with $\frac{l(\kappa)}{2}$-bit,
where $l(\kappa)=4\kappa$.
Specifically, $n(\kappa)$ and $l(\kappa)$ denote the bit length of the modulus $N$ and the private key, respectively, corresponding to $\kappa$-bit security level.
Besides, it satisfies that $p|(P-1)$, $q|(Q-1)$, $P\equiv Q \equiv 3 \mod 4$, $\gcd(P-1,Q-1)=2$ and $\gcd(pq,(P-1)(Q-1)/(4pq))=1$.


\subsubsection{\textbf{Key generation} (\texttt{KeyGen})}\label{KeyGen of optimized Paillier}
\texttt{KeyGen} takes a security parameter $\kappa$ as input, and outputs a pair of private key and public key $(sk,pk)$.
\texttt{KeyGen} initially invokes \texttt{NGen} to retrieve a tuple $(N,P,Q,p,q)$, subsequently computes $\alpha = pq$ and $\beta = \frac{(P-1)(Q-1)}{4pq}$.
Then, \texttt{KeyGen} calculates $h = -y^{2\beta} \!\!\! \mod N$, where $y\leftarrow \mathbb{Z}_{N}^{*}$.
Finally, \texttt{KeyGen} outputs the public key $pk = (N,h)$ and the private key $sk = \alpha$.

\subsubsection{\textbf{Encryption} \texttt{(Enc)}}\label{Enc of optimized Paillier}
\texttt{Enc} takes a message $m \in {\mathbb{Z}}_{N}$ and a public key $pk = (N,h)$ as inputs and outputs a ciphertext $c \in {\mathbb{Z}}^{*}_{{N}^{2}}$. 
Formally, \texttt{Enc} is defined as
\begin{align}\label{encryption____________________________}
&c \leftarrow \texttt{Enc}(pk,m)={(1+N)}^m\cdot {(h^r \!\!\!\!\!\! \mod N)}^N  \!\!\!\!\!\! \mod N^2. 
\end{align}

In Eq. (\ref{encryption____________________________}), $r \leftarrow {\{0,1\}}^{l(\kappa)}$.
In this work, the notation $\llbracket x \rrbracket $ denotes an encryption of the value $x$.

\subsubsection{\textbf{Decryption} \texttt{(Dec)}}\label{Dec of optimized Paillier}
\texttt{Dec} takes a ciphertext $c \in {\mathbb{Z}}^{*}_{{N}^{2}}$ and a private key $sk = \alpha$ as inputs and outputs a plaintext message $m \in {\mathbb{Z}}_{N}$.
Formally, \texttt{Dec} is defined as
\begin{align}\label{decrytion_____aaaaaaaaaaaaaaaaaaaa_____}
&m \leftarrow \texttt{Dec}(sk,c)=L(c^{2\alpha}\ \!\!\!\!\!\! \mod\ N^2) \cdot{(2\alpha)}^{-1}\ \!\!\!\!\!\! \mod N.
\end{align}

In Eq. (\ref{decrytion_____aaaaaaaaaaaaaaaaaaaa_____}), the notation $L(x)$ denotes $\frac{x-1}{N}\mod N$.

FastPai supports additive homomorphism and scalar-multiplication homomorphism as demonstrated below.
\begin{itemize}
    \item 
    $\texttt{Dec}(sk,\llbracket m_1 \rrbracket \cdot \llbracket m_2 \rrbracket) = \texttt{Dec}(sk,\llbracket m_1+m_2 \rrbracket)$.
    \item
    $\texttt{Dec}(sk,{\llbracket m \rrbracket}^r) = \texttt{Dec}(sk,\llbracket r\cdot m \rrbracket)$, where $r$ is a constant. When $r = N-1$, ${\texttt{Dec}(sk,{\llbracket m \rrbracket}^r)} = -m$ holds.
\end{itemize}

Besides, to enable operations on a negative number, we convert a negative number $m$ into $N - \left | m \right | $ in this work. 
In this case, $[0,\frac{N}{2}]$ and $[\frac{N}{2}+1, N-1]$ denote the space of natural number and the one of negative number, respectively.

\subsection{Secret Sharing}
The scheme OSY21 \cite{orlandi2021rise} utilizes secret sharing to partition a secret $x$ into two shares ${\left \langle x \right \rangle}_0$ and ${\left \langle x \right \rangle}_1$, such that ${\left \langle x \right \rangle}_1-{\left \langle x \right \rangle}_0 = x$.
Specifically, for $x $ within the range $ [0,m-1]$, ${\left \langle x \right \rangle}_1$ is randomly chosen from the interval $[0,m2^\kappa-1]$, with $\kappa$ representing a statistical security parameter. 
${\left \langle x \right \rangle}_0$ is then defined as ${\left \langle x \right \rangle}_1-x$.

In this work, we employ a technique to facilitate operations on negative number, detailed as follows.
Assuming $x\in(-m,m)$, we select ${\left \langle x \right \rangle}_1$ from the interval $ [-m\cdot2^\kappa+1,m\cdot2^\kappa-1]$, and define ${\left \langle x \right \rangle}_0 = {\left \langle x \right \rangle}_1 - x$.

Since a secret data is recovered by subtraction, as same as OSY21 \cite{orlandi2021rise}, we call the above secret sharing as subtractive secret sharing in this work, and call the shares as subtractive shares.

\subsection{Share Conversion Algorithm} \label{DDLog}
The scheme OSY21 \cite{orlandi2021rise} proposes a share conversion algorithm called \texttt{DDLog\textsubscript{N}}$(g)$ to solve the distributed discrete log problem.
The purpose of \texttt{DDLog}\textsubscript{N}$(g)$ is to convert divisive shares into subtractive shares.
Specifically, let $g_0,g_1\in \mathbb{Z}_{N}^{*}$ be held by two different parties, satisfying $g_1 = g_0\cdot(1+N)^x \mod N^2$. 
Given $z_i = \texttt{DDLog\textsubscript{N}}(g_i) (i\in \{0,1\})$ held by two different parties, it follows that $z_1 - z_0 = x \mod N$. 
$g_1$ and $g_0$ are defined as divisive shares of $x$ since $\frac{g_1}{g_0} = 1+x\cdot N \mod N^2$, while $z_1$ and $z_0$ are defined as subtractive shares of $x$ since $z_1 - z_0 = x \mod N$.
Specifically, \texttt{DDLog}\textsubscript{N}$(g)$ consists of the following steps.
\begin{enumerate} [(a)]
    \item Define $h=g \!\! \mod N$ and $h' = \left \lfloor g/N \right \rfloor $.
    \item Output $z=h'\cdot h^{-1} \mod N$.
\end{enumerate}

In this work, we adopt \texttt{DDLog}\textsubscript{N}$(g)$ to construct MORSE.

\subsection{Key-Dependent Messages (KDM) Security}
If an adversary obtains the encrypted values of certain function $g(K)$ with respect to the key vector $K$ under a specific key $K_i$, and the encryption scheme maintains security in terms of indistinguishability, it is said to possess KDM security \cite{black2003encryption}.
Let $\mathcal{F}$ represent a set of functions with respect to secret keys, the work \cite{brakerski2010circular} defines the KDM\textsuperscript{($n$)} game between a challenger and a polynomial time adversary $\mathcal{A}$ as follows.
\begin{enumerate}[(a)]
    \item Initialization. The challenger randomly selects $b$ from $\{0,1\}$, and generates $n$ ($n>0$) pairs of keys, where $i$-th pair of key is denoted as ${\{pk_i, sk_i\}}_{1\leq i \leq n}$.
    Next, the challenger sends all ${\{pk_i\}}_{1\leq i \leq n}$ to the adversary.
    \item Query. The adversary makes queries of $(i, f)$, where $1\leq i \leq n$ and $f \in \mathcal{F}$. 
    For each query, the challenger sends the following ciphertext $c$ to the adversary. 
    If $b=0$, $c\leftarrow \texttt{Enc}(pk_i,f(sk_1,...,sk_n))$, otherwise ($b=1$), $c\leftarrow \texttt{Enc}(pk_i,0)$.
    \item Finish. The adversary outputs a guess $b' \in \{0,1\}$.
\end{enumerate}

The advantage for the adversary $\mathcal{A}$ is defined as
\begin{align} \label{advantage}
    &\advantage[\mathcal{A}]{\text{KDM\textsuperscript{($n$)}}} \nonumber \\
    &=|\prob{b'=1|b=0}-\prob{b'=1|b=1}|.
\end{align}

A public key encryption scheme is KDM\textsuperscript{($n$)} secure if $\advantage[\mathcal{A}]{\text{KDM\textsuperscript{($n$)}}}$ is negligible.
Inspired by the work OSY21 \cite{orlandi2021rise}, this work only considers one key pair (i.e., $n=1$).

\begin{figure}[!t]
    \centering
    \resizebox{0.47\textwidth}{!}{
    \includegraphics{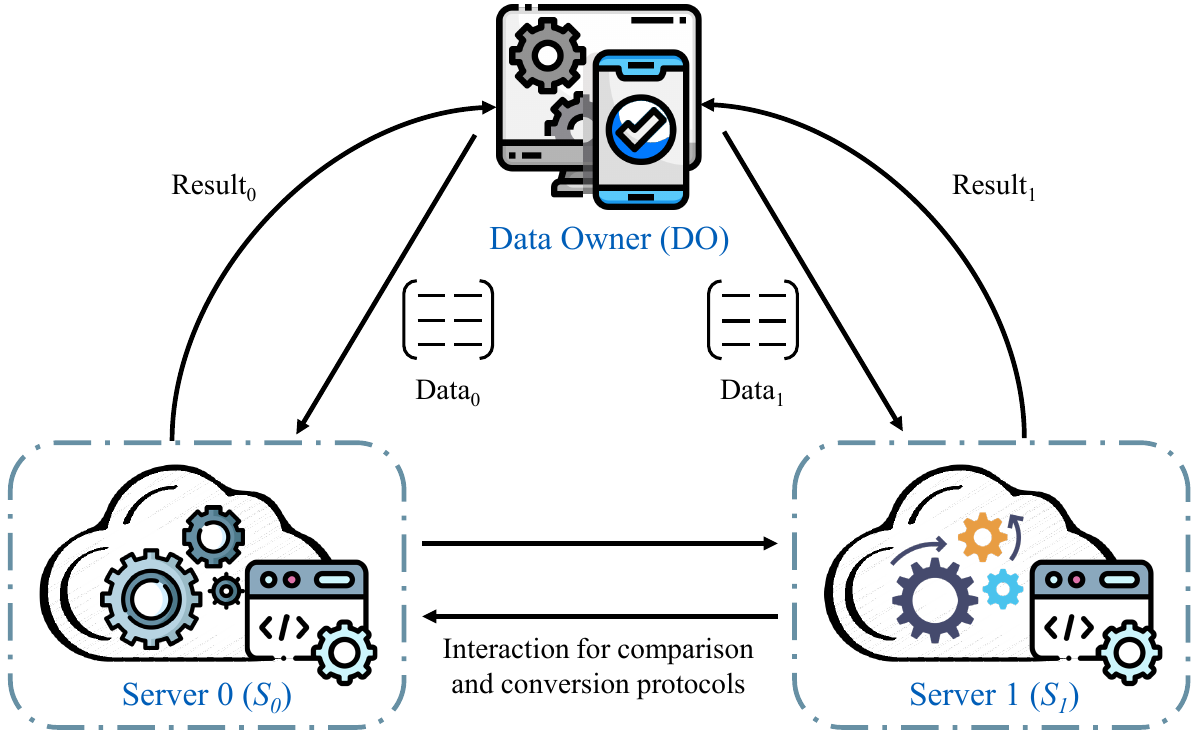}
    }
    \caption{The system model of MORSE}
    \label{system_model}
\end{figure}
\section{System model and threat model}\label{Section_4}

\subsection{System Model}
As depicted in Fig. \ref{system_model}, the system model of MORSE consists of a data owner, a server 0 and a server 1, and the two servers are denoted as $S_0$ and $S_1$, respectively.

\begin{itemize}
    \item \textbf{Data Owner (DO)}. 
    DO initializes FastPai \cite{ma2021optimized} cryptosystem, retains the private key $sk$, and publicly discloses the public key $pk$. 
    DO uploads data to $S_0$ and $S_1$, and receives the shares results from both $S_0$ and $S_1$ or ciphertexts results from $S_0$.
    \item \textbf{$\bm{S_0}$ and $\bm{S_1}$}.
    $S_0$ and $S_1$ provide storage and computing services for DO. 
    Specifically, for an uploaded data $x$, both $S_0$ and $S_1$ hold the same Paillier ciphertext $\llbracket x \rrbracket$, and they respectively hold the shares ${\left \langle 2\alpha\cdot x \right \rangle}_0$ and ${\left \langle 2\alpha\cdot x \right \rangle}_1$.
    They locally complete linear operations, and collaboratively perform non-linear operations and the conversions between Paillier ciphertexts and shares through interactions.
    $S_0$ can send the results in the form of both Paillier ciphertexts and shares to DO, but $S_1$ can only send the results in the form of shares to DO.
\end{itemize}

\subsection{Threat Model}
DO is considered as fully trusted in MORSE, and the potential adversaries may include $S_0$ and $S_1$.
In MORSE, $S_0$ and $S_1$ are regarded as semi-honest, i.e., they will correctly perform the requested protocols, but attempt to infer the private information of DO.  
We assume $S_0$ and $S_1$ are non-colluding, an assumption that is widely adopted in the schemes \cite{hu2023achieving}, \cite{huang2019lightweight} and \cite{xie2021achieving}. 
The non-colluding assumption implies that $S_0$ and $S_1$ do not disclose any information to each other unless the information needed for the execution of protocols.

Considering $S_0$ and $S_1$ originate from different cloud service providers, the non-colluding assumption is practical \cite{zhao2022soci}, \cite{xie2021achieving}.
In MORSE, the collusion between $S_0$ and $S_1$ means revealing their own shares of data to one another. 
If $S_0$ reveals its shares, it provides $S_1$ with evidence of data leakage, subsequently $S_1$ could invoke data privacy law to penalize $S_0$ and furthermore capture the market shares of $S_0$.
Conversely, if $S_1$ discloses its shares, $S_0$ can initiate similar measures.
Therefore, for the sake of their respective commercial interests, the two servers will not collude.

\section{MORSE Design} \label{Section_5}
\subsection{MORSE Initialization}
Before performing linear and non-linear operations, DO uploads data to both $S_0$ and $S_1$.
Specifically, for a data $x \in [-2^l,2^l]$ (e.g., $l=32$), DO performs the following operations.
    \begin{enumerate}[(a)]
        \item DO invokes \texttt{Enc} to encrypt $x$ as $\llbracket x \rrbracket$, and employs secret sharing to partition $2\alpha\cdot x$ into ${\left \langle 2\alpha\cdot x \right \rangle}_0$ and ${\left \langle 2\alpha\cdot x \right \rangle}_1$, where $\alpha$ is the private key of FastPai. 
        \item DO sends $\llbracket x \rrbracket$ to both $S_0$ and $S_1$.
        \item DO sends ${\left \langle 2\alpha\cdot x \right \rangle}_0$ and ${\left \langle 2\alpha\cdot x \right \rangle}_1$ to $S_0$ and $S_1$, respectively.
    \end{enumerate}
    
    Furthermore, to ensure MORSE correctly outputting the results, DO is required to perform the following operations.
    \begin{enumerate}[(a)]
        \item DO invokes \texttt{Enc} to encrypt ${(2\alpha)}^{-1} \mod N$, $0$ and $1$ into $\llbracket {(2\alpha)}^{-1} \rrbracket$, $\llbracket 0 \rrbracket$ and $\llbracket 1 \rrbracket$, respectively.
        \item 
        DO adopts secret sharing to split $2\alpha$ into ${\left \langle 2\alpha \right \rangle}_0$ and ${\left \langle 2\alpha \right \rangle}_1$.
        \item 
        DO sends $\llbracket {(2\alpha)}^{-1} \rrbracket$, $\llbracket 0 \rrbracket$, $\llbracket 1 \rrbracket$ to both $S_0$ and $S_1$, and sends ${\left \langle 2\alpha \right \rangle}_0$ and ${\left \langle 2\alpha \right \rangle}_1$ to $S_0$ and $S_1$, respectively.
    \end{enumerate}

    After receiving the data uploaded by DO, $S_0$ and $S_1$ construct tuples $Assisted_{S_0}$ and $Assisted_{S_1}$, respectively, where $Assisted_{S_0}=\{{\left \langle 2\alpha \right \rangle}_0, \llbracket {(2\alpha)}^{-1}  \rrbracket, \llbracket 0 \rrbracket, \llbracket 1 \rrbracket\}$ and $Assisted_{S_1}=\{{\left \langle 2\alpha \right \rangle}_1, \llbracket {(2\alpha)}^{-1} \rrbracket, \llbracket 0 \rrbracket, \llbracket 1 \rrbracket\}$.

\subsection{Linear Operation}
\label{SMUL}
MORSE enables linear operations including secure multiplication (\texttt{SMUL\textsubscript{HSS}}), secure addition, secure subtraction and scalar-multiplication.
In this subsection, we elaborate on secure multiplication.

As shown in Algorithm \ref{Secure Multiplication}, to perform secure multiplication between $x$ and $y$, \texttt{SMUL\textsubscript{HSS}} requires both $S_0$ and $S_1$ to input the same Paillier ciphertext $\llbracket x \rrbracket$.
Besides, $S_0$ and $S_1$ are required to input ${\left \langle 2\alpha \cdot y \right \rangle}_0$ and ${\left \langle 2\alpha \cdot y \right \rangle}_1$, respectively.
Finally, $S_0$ and $S_1$ obtain $z_0 = {\left \langle 2\alpha \cdot x \cdot y \right \rangle}_0$ and $z_1 = {\left \langle 2\alpha \cdot x \cdot y \right \rangle}_1$, respectively, such that $z_1 - z_0 = 2\alpha \cdot x \cdot y \mod N $.
Specifically, \texttt{SMUL\textsubscript{HSS}} consists of two steps as follows.
\begin{enumerate}
    \item $S_0$ and $S_1$ compute $g_0 \leftarrow {\llbracket x \rrbracket}^{{\left \langle 2\alpha \cdot y \right \rangle}_0 }$ and $g_1 \leftarrow {\llbracket x \rrbracket}^{{\left \langle 2\alpha \cdot y \right \rangle}_1 }$, respectively, such that $g_1 = g_0\cdot {\llbracket x \rrbracket}^{2\alpha \cdot y} \mod N^2$.
    Since $\frac{g_1}{g_0}={\llbracket x \rrbracket}^{2\alpha \cdot y} \mod N^2$, $g_0$ and $g_1$ refer to the divisive shares of $2\alpha\cdot x\cdot y$.
    \item  $S_0$ and $S_1$ convert the divisive shares into subtractive shares. 
    Specifically, $S_0$ and $S_1$ compute $z_0\leftarrow \texttt{DDLog\textsubscript{N}}(g_0)$ and $z_1\leftarrow \texttt{DDLog\textsubscript{N}}(g_1)$, respectively, such that $z_1 - z_0 = 2\alpha \cdot x \cdot y \mod N $.
\end{enumerate}

After receiving $z_0$ and $z_1$ from $S_0$ and $S_1$, respectively, DO obtains $x\cdot y$ by computing $(z_1-z_0)\cdot {(2\alpha)}^{-1} = x\cdot y \mod N$.

\begin{algorithm}[!ht]
    \caption{Secure Multiplication (\texttt{SMUL\textsubscript{HSS}})}
    \label{Secure Multiplication}
    \KwIn{$S_0$ inputs ($\llbracket x \rrbracket $, ${\left \langle 2\alpha \cdot y \right \rangle}_0$), and $S_1$ inputs ($\llbracket x \rrbracket $, ${\left \langle 2\alpha \cdot y \right \rangle}_1$).}
    \KwOut{$S_0$ outputs $z_0 = {\left \langle 2\alpha \cdot x \cdot y \right \rangle}_0$, and $S_1$ outputs $z_1= {\left \langle 2\alpha \cdot x \cdot y \right \rangle}_1$, such that $z_1 - z_0 = 2\alpha \cdot x \cdot y \mod N $.}
    \vspace{0.5em}
    \begin{mdframed}[backgroundcolor=cpcolor, innerleftmargin=2pt,innerrightmargin=2pt,innerbottommargin=2pt]
                    $\triangleright$ Step (\romannumeral1). $S_0$ and $S_1$ obtain divisive shares:
                    \begin{mdframed}[backgroundcolor=whilecolor,innertopmargin=4pt,innerbottommargin=4pt]
                    \begin{algorithmic}[1]
                        \item $S_0$ and $S_1$ compute $g_0 \leftarrow {\llbracket x \rrbracket}^{{\left \langle 2\alpha \cdot y \right \rangle}_0 }$ and $g_1 \leftarrow {\llbracket x \rrbracket}^{{\left \langle 2\alpha \cdot y \right \rangle}_1 }$, respectively, such that $g_1 = g_0\cdot {\llbracket x \rrbracket}^{2\alpha\cdot y} \mod N^2$.
                    \end{algorithmic}
                    \end{mdframed}
    \end{mdframed}
    \begin{mdframed}[backgroundcolor=cspcolor, innerleftmargin=2pt,innerrightmargin=2pt,innerbottommargin=2pt]
                    $\triangleright$ Step (\romannumeral2). 
                    $S_0$ and $S_1$ convert divisive shares into subtractive shares using \texttt{DDLog\textsubscript{N}}$(g)$:
                    \begin{mdframed}[backgroundcolor=whilecolor,innertopmargin=4pt,innerbottommargin=4pt]
                    \begin{algorithmic}[1]
                        \item $S_0$ and $S_1$ compute $z_0 \leftarrow \texttt{DDLog\textsubscript{N}}(g_0)$ and $z_1 \leftarrow \texttt{DDLog\textsubscript{N}}(g_1)$, respectively, such that $z_1 - z_0 = 2\alpha \cdot x \cdot y \mod N $.
                    \end{algorithmic}
                    \end{mdframed}
    \end{mdframed}
\end{algorithm}

In \texttt{SMUL\textsubscript{HSS}}, the most time-consuming operation is the modular exponential operation, where the exponent is a share involving the private key of FastPai.
Compared to  FGJS17 \cite{fazio2017homomorphic}, OSY21 \cite{orlandi2021rise} and RS21 \cite{roy2021large}, MORSE enjoys smallest share, which leads to faster modular exponential operation.
Besides, different from FGJS17 \cite{fazio2017homomorphic} and OSY21 \cite{orlandi2021rise}, since the shares in MORSE contain the private key of FastPai, MORSE does not require partitioning the private key into multiple segments and generating ciphertexts like $\llbracket sk_1 \cdot x \rrbracket, \llbracket sk_2 \cdot x \rrbracket, ..., \llbracket sk_l \cdot x \rrbracket$, where $l$ denotes the number of segments the private key is divided into.
Moreover, even the DDLog in OSY21 \cite{orlandi2021rise} enjoys perfect correctness, OSY21 \cite{orlandi2021rise} suffers from a negligible correctness error since it converts shares over $N$ into shares over integer.
MORSE discards this conversion to avoid potential computation error.

\label{modulus_reduction}
Assuming the multiplication $xy$ of $x$ and $y$ is at the interval $[0,2^l]$ (e.g., $l=32$), we can transform the modulus of $z_0$ and $z_1$ from $N$ to a smaller modulus $\theta$ to make next secure multiplication more efficient.
In this case, we simply let $\theta=2^{4\kappa+1+\kappa}$, where $\kappa$ is a security parameter.
After obtaining the outputs of \texttt{SMUL\textsubscript{HSS}}, i.e., $z_0$ and $z_1$, $S_0$ and $S_1$ compute $z_0 \leftarrow z_0 \mod \theta$ and $z_1 \leftarrow z_1 \mod \theta$, respectively.
Next, DO computes $z'\leftarrow z_1-z_0 \mod \theta$ and obtains $xy$ by computing $xy\leftarrow \frac{z'}{2\alpha}$.
Since $2\alpha\cdot xy < \theta$ and $\theta < N$, the correctness can be easily verified.

In MORSE, secure addition, secure subtraction and scalar-multiplication are supported by both FastPai \cite{ma2021optimized} and secret sharing.
For ciphertexts, the operations are performed as follows.
\begin{itemize}
    \item Addition.
    For ciphertexts $\llbracket x \rrbracket$ and $\llbracket y \rrbracket$, we compute $\llbracket x+y \rrbracket \leftarrow \llbracket x \rrbracket \cdot \llbracket y \rrbracket \mod N^2$.
    \item Subtraction.
    For ciphertexts $\llbracket x \rrbracket$ and $\llbracket y \rrbracket$, we compute $\llbracket x-y \rrbracket \leftarrow \llbracket x \rrbracket \cdot {\llbracket y \rrbracket}^{-1} \mod N^2$.
    \item Scalar-multiplication.
    For a ciphertext $\llbracket x \rrbracket$ and a constant $c$, we compute $\llbracket c\cdot x \rrbracket \leftarrow {\llbracket x \rrbracket}^{c}  \mod N^2$.
\end{itemize}

For shares, the operations are performed as follows.
\begin{itemize}
    \item Addition.
    For shares (${\left \langle x \right \rangle}_0$,${\left \langle y \right \rangle}_0$) and (${\left \langle x \right \rangle}_1$,${\left \langle y \right \rangle}_1$), we compute
    ${\left \langle x+y \right \rangle}_0={\left \langle x \right \rangle}_0+{\left \langle y \right \rangle}_0$ and ${\left \langle x+y \right \rangle}_1={\left \langle x \right \rangle}_1+{\left \langle y \right \rangle}_1$.
    \item Subtraction.
    For shares (${\left \langle x \right \rangle}_0$,${\left \langle y \right \rangle}_0$) and (${\left \langle x \right \rangle}_1$,${\left \langle y \right \rangle}_1$), we compute
    ${\left \langle x-y \right \rangle}_0={\left \langle x \right \rangle}_0-{\left \langle y \right \rangle}_0$ and ${\left \langle x-y \right \rangle}_1={\left \langle x \right \rangle}_1-{\left \langle y \right \rangle}_1$.
    \item Scalar-multiplication.
    For shares (${\left \langle x \right \rangle}_0$, ${\left \langle x \right \rangle}_1$) and a constant $c$, we compute
    ${\left \langle c\cdot x \right \rangle}_0=c\cdot {\left \langle x \right \rangle}_0$ and ${\left \langle c\cdot x \right \rangle}_1=c \cdot {\left \langle x \right \rangle}_1$. 
\end{itemize}


\subsection{Non-Linear Operation}
\label{SCMP}
The previous schemes FGJS17 \cite{fazio2017homomorphic}, OSY21 \cite{orlandi2021rise} and RS21 \cite{roy2021large} fail to support non-linear operation.
In this subsection, we present a secure comparison protocol (\texttt{SCMP\textsubscript{HSS}}).
As shown in Algorithm \ref{Secure Comparison}, in \texttt{SCMP\textsubscript{HSS}}, $S_0$ inputs $\llbracket x \rrbracket $ and $\llbracket y \rrbracket $, and only $S_0$ obtains the output $\llbracket \mu \rrbracket$. If $x\geq y$, $\mu=0$, otherwise ($x < y$), $\mu=1$.
Specifically, \texttt{SCMP\textsubscript{HSS}} consists of the following steps.
\begin{enumerate}
    \item $S_0$ randomly selects $r_1$ and $r_2$, where $r_1\leftarrow {\{0,1\}}^\sigma \setminus \{0\} (e.g., \sigma=128)$, $ r_2 \leq \frac{N}{2}$ and $r_1 + r_2 > \frac{N}{2}$. Next, $S_0$ takes out ${\left \langle 2\alpha \right \rangle}_0 ,\llbracket 0 \rrbracket $ and $\llbracket 1\rrbracket$ from $Assisted_{S_0}$.
    Subsequently, $S_0$ selects $\pi \leftarrow \{0,1\}$.
    If $\pi = 0$, $S_0$ computes $D\leftarrow {(\llbracket x \rrbracket \cdot {\llbracket  y \rrbracket }^{-1} \cdot \llbracket 1 \rrbracket) }^{r_1}$, otherwise ($\pi = 1$), computes $D\leftarrow {(\llbracket  y \rrbracket \cdot {\llbracket x \rrbracket }^{-1})}^{r_1}$.
    Then, $S_0$ computes $z_0\leftarrow \texttt{DDLog\textsubscript{N}}(D^{{\left \langle 2\alpha \right \rangle}_0})$ and $z_0 \leftarrow z_0-r_2$.
    Finally, $S_0$ sends $(D,z_0)$ to $S_1$.
    \item 
    $S_1$ takes out ${\left \langle 2\alpha \right \rangle}_1, \llbracket 0 \rrbracket $ and $\llbracket 1\rrbracket$ from $Assisted_{S_1}$.
    Subsequently, $S_1$ computes $z_1\leftarrow \texttt{DDLog\textsubscript{N}}(D^{{\left \langle 2\alpha \right \rangle}_1})$ and $d \leftarrow z_1-z_0 \mod N$.
    If $\pi = 0$, $d = 2\alpha\cdot r_1 \cdot (x-y+1)+r_2$, otherwise ($\pi=1$), $d = 2\alpha\cdot r_1 \cdot (y-x)+r_2$.
    If $d >\frac{N}{2}$, $S_1$ computes $\llbracket {\mu}_0 \rrbracket \leftarrow \llbracket 0 \rrbracket $, otherwise ($d \leq \frac{N}{2}$), $S_1$ computes $\llbracket {\mu}_0 \rrbracket \leftarrow \llbracket 1 \rrbracket $.
    Finally, $S_1$ sends $\llbracket {\mu}_0 \rrbracket$ to $S_0$.
    \item 
    If $\pi = 0$, $S_0$ computes $\llbracket \mu \rrbracket  \leftarrow \llbracket {\mu}_0 \rrbracket \cdot \llbracket 0 \rrbracket$, otherwise ($\pi = 1$), $S_0$ computes $\llbracket \mu \rrbracket  \leftarrow \llbracket 1 \rrbracket \cdot {\llbracket {\mu}_0\rrbracket }^{-1}$.    
    Finally, $S_0$ computes $\llbracket \mu \rrbracket \leftarrow \llbracket \mu \rrbracket \cdot \llbracket 0 \rrbracket$.
\end{enumerate}

\begin{algorithm}[!ht]
    \caption{Secure Comparison (\texttt{SCMP\textsubscript{HSS}})}
    \label{Secure Comparison}
    \KwIn{$S_0$ inputs $\llbracket x \rrbracket $ and $\llbracket y \rrbracket $.}
    \KwOut{$S_0$ outputs $\llbracket \mu \rrbracket$. If $x\geq y$, $\mu=0$, otherwise ($x < y$), $\mu=1$.}
    \vspace{0.5em}
    \begin{mdframed}[backgroundcolor=cpcolor, innerleftmargin=2pt,innerrightmargin=2pt,innerbottommargin=2pt]
                    $\triangleright$ Step (\romannumeral1). 
                    $S_0$ performs the following operations.
                    \begin{mdframed}[backgroundcolor=whilecolor,innertopmargin=4pt,innerbottommargin=4pt]
                    \begin{algorithmic}[1]
                    \item $S_0$ randomly selects $r_1$ and $r_2$, where $r_1\leftarrow {\{0,1\}}^\sigma \setminus \{0\} (e.g., \sigma=128)$, $ r_2 \leq \frac{N}{2}$ and $r_1 + r_2 > \frac{N}{2}$.
                    \item $S_0$ takes out ${\left \langle 2\alpha \right \rangle}_0 ,\llbracket 0 \rrbracket $ and $\llbracket 1\rrbracket$ from $Assisted_{S_0}$.
                    \item $S_0$ selects $\pi \leftarrow \{0,1\}$.
                    If $\pi = 0$, $S_0$ computes $D\leftarrow {(\llbracket x \rrbracket \cdot {\llbracket  y \rrbracket }^{-1} \cdot \llbracket 1 \rrbracket) }^{r_1}$, otherwise ($\pi = 1$), computes $D \leftarrow {(\llbracket  y \rrbracket \cdot {\llbracket x \rrbracket }^{-1})}^{r_1}$.
                    \item 
                    $S_0$ computes $z_0 \leftarrow \texttt{DDLog\textsubscript{N}}(D^{{\left \langle 2\alpha \right \rangle}_0})$ and $z_0 \leftarrow z_0-r_2$.
                    \item 
                    $S_0$ sends $(D,z_0)$ to $S_1$.
                    \end{algorithmic}
                    \end{mdframed}
    \end{mdframed}
    \begin{mdframed}[backgroundcolor=cspcolor, innerleftmargin=2pt,innerrightmargin=2pt,innerbottommargin=2pt]
                    $\triangleright$ Step (\romannumeral2). 
                    $S_1$ performs the following operations.
                    \begin{mdframed}[backgroundcolor=whilecolor,innertopmargin=4pt,innerbottommargin=4pt]
                    \begin{algorithmic}[1]
                    \item $S_1$ takes out ${\left \langle 2\alpha \right \rangle}_1, \llbracket 0 \rrbracket $ and $\llbracket 1\rrbracket$ from $Assisted_{S_1}$.
                    \item 
                    $S_1$ computes $z_1 \leftarrow \texttt{DDLog\textsubscript{N}}(D^{{\left \langle 2\alpha \right \rangle}_1})$ and $d \leftarrow z_1-z_0 \mod N$.
                    \item 
                    If $d >\frac{N}{2}$, $S_1$ computes $\llbracket {\mu}_0 \rrbracket \leftarrow \llbracket 0 \rrbracket $, otherwise ($d \leq \frac{N}{2}$), $S_1$ computes $\llbracket {\mu}_0 \rrbracket \leftarrow \llbracket 1 \rrbracket $.
                    \item 
                    $S_1$ sends $\llbracket {\mu}_0 \rrbracket$ to $S_0$.
                    \end{algorithmic}
                    \end{mdframed}
    \end{mdframed}
        \begin{mdframed}[backgroundcolor=cpcolor, innerleftmargin=2pt,innerrightmargin=2pt,innerbottommargin=2pt]
                    $\triangleright$ Step (\romannumeral3).
                    $S_0$ performs the following operations.
                    \begin{mdframed}[backgroundcolor=whilecolor,innertopmargin=4pt,innerbottommargin=4pt]
                    \begin{algorithmic}[1]
                   \item If $\pi = 0$, $S_0$ computes $\llbracket \mu \rrbracket  \leftarrow \llbracket {\mu}_0 \rrbracket$, otherwise ($\pi = 1$), $S_0$ computes $\llbracket \mu \rrbracket  \leftarrow \llbracket 1 \rrbracket \cdot {\llbracket {\mu}_0\rrbracket }^{-1}$.
                   \item 
                   $S_0$ computes $\llbracket \mu \rrbracket \leftarrow \llbracket \mu \rrbracket \cdot \llbracket 0 \rrbracket$.
                    \end{algorithmic}
                    \end{mdframed}
    \end{mdframed}
\end{algorithm}

\subsection{Conversion protocols} \label{conversion protocol}
The outputs of \texttt{SMUL\textsubscript{HSS}} are in the form of shares, whereas the inputs for \texttt{SCMP\textsubscript{HSS}} 
are in the form of ciphertexts, thus the outputs of \texttt{SMUL\textsubscript{HSS}} can not serve as the inputs for \texttt{SCMP\textsubscript{HSS}}.
Additionally, \texttt{SMUL\textsubscript{HSS}} accepts inputs in two forms, i.e., shares and ciphertexts.
However, since the output of \texttt{SCMP\textsubscript{HSS}} is solely in the form of ciphertexts, it can not be used as a share-input for \texttt{SMUL\textsubscript{HSS}}.
To address these challenges, we introduce two conversion protocols, i.e., a share to ciphertext (\texttt{S2C}) protocol and a ciphertext to share (\texttt{C2S}) protocol. 
One for converting shares into ciphertexts, and another for converting ciphertexts into shares.
With \texttt{S2C} and \texttt{C2S}, MORSE achieves continuous computations between \texttt{SMUL\textsubscript{HSS}} and \texttt{SCMP\textsubscript{HSS}}.

Assuming $S_0$ and $S_1$ have the outputs of \texttt{SMUL\textsubscript{HSS}}, i.e., ${\left \langle 2\alpha \cdot x \cdot y  \right \rangle}_0$ and ${\left \langle 2\alpha \cdot x \cdot y  \right \rangle}_1$, respectively, by using \texttt{S2C} protocol, both $S_0$ and $S_1$ obtain the same ciphertext $\llbracket x\cdot y \rrbracket$.
Specifically, \texttt{S2C} consists of the following steps.
\begin{enumerate}
    \item $S_0$ takes out $\llbracket {(2\alpha)}^{-1} \rrbracket$ from $Assisted_{S_0}$, and computes ${\left \langle {(2\alpha)}^{-1}\cdot (2\alpha) \cdot x \cdot y  \right \rangle}_0 \leftarrow \texttt{DDLog\textsubscript{N}}({\llbracket {(2\alpha)}^{-1} \rrbracket}^{{\left \langle 2\alpha \cdot x \cdot y  \right \rangle}_0})$.
    Subsequently, $S_0$ invokes \texttt{Enc} to encrypt ${\left \langle {(2\alpha)}^{-1}\cdot (2\alpha) \cdot x \cdot y  \right \rangle}_0$ into $\llbracket {\left \langle {(2\alpha)}^{-1}\cdot (2\alpha) \cdot x \cdot y  \right \rangle}_0 \rrbracket$, and sends $\llbracket {\left \langle {(2\alpha)}^{-1}\cdot (2\alpha) \cdot x \cdot y  \right \rangle}_0 \rrbracket$ to $S_1$.
    \item $S_1$ takes out $\llbracket {(2\alpha)}^{-1} \rrbracket$ from $Assisted_{S_1}$, and computes ${\left \langle {(2\alpha)}^{-1}\cdot (2\alpha) \cdot x \cdot y  \right \rangle}_1 \leftarrow \texttt{DDLog\textsubscript{N}}({\llbracket {(2\alpha)}^{-1} \rrbracket}^{{\left \langle 2\alpha \cdot x \cdot y  \right \rangle}_1})$.
    Subsequently, $S_1$ invokes \texttt{Enc} to encrypt ${\left \langle {(2\alpha)}^{-1}\cdot (2\alpha) \cdot x \cdot y  \right \rangle}_1$ into $\llbracket {\left \langle {(2\alpha)}^{-1}\cdot (2\alpha) \cdot x \cdot y  \right \rangle}_1 \rrbracket$, and sends $\llbracket {\left \langle {(2\alpha)}^{-1}\cdot (2\alpha) \cdot x \cdot y  \right \rangle}_1 \rrbracket$ to $S_0$.
    \item Both $S_0$ and $S_1$ compute $\llbracket x\cdot y \rrbracket \leftarrow \llbracket {\left \langle {(2\alpha)}^{-1}\cdot (2\alpha) \cdot x \cdot y  \right \rangle}_1 \rrbracket \cdot {\llbracket {\left \langle {(2\alpha)}^{-1}\cdot (2\alpha) \cdot x \cdot y  \right \rangle}_0 \rrbracket}^{-1}$.
\end{enumerate}

Assuming $S_0$ has the output of \texttt{SCMP\textsubscript{HSS}}, i.e., $\llbracket \mu \rrbracket$, by using $\texttt{C2S}$ protocol, $S_0$ and $S_1$ obtain 
${\left \langle 2\alpha \cdot \mu  \right \rangle}_0$ and ${\left \langle 2\alpha \cdot \mu  \right \rangle}_1$, respectively, such that ${\left \langle 2\alpha \cdot \mu  \right \rangle}_1 - {\left \langle 2\alpha \cdot \mu  \right \rangle}_0 = 2\alpha \cdot \mu \mod N$.
Specifically, $\texttt{C2S}$ consists of the following steps.
\begin{enumerate}
    \item $S_0$ sends $\llbracket \mu \rrbracket$ to $S_1$.
    \item $S_0$ takes out ${\left \langle 2\alpha \right \rangle}_0$ from $Assisted_{S_0}$ and computes ${\left \langle 2\alpha \cdot \mu  \right \rangle}_0 \leftarrow \texttt{DDLog\textsubscript{N}}({\llbracket \mu \rrbracket}^{{\left \langle 2\alpha \right \rangle}_0})$.
    \item $S_1$ takes out ${\left \langle 2\alpha \right \rangle}_1$ from $Assisted_{S_1}$ and computes ${\left \langle 2\alpha \cdot \mu  \right \rangle}_1 \leftarrow \texttt{DDLog\textsubscript{N}}({\llbracket \mu \rrbracket}^{{\left \langle 2\alpha \right \rangle}_1})$. 
\end{enumerate}

\section{Correctness and Security Analyses}\label{Section_6}
In this Section, through rigorous analyses, we show that MORSE securely outputs correct results.
\subsection{Correctness Analysis}
\begin{theorem} \label{theorem_SMUL}
Given ($\llbracket x \rrbracket $, ${\left \langle 2\alpha \cdot y \right \rangle}_0$) and ($\llbracket x \rrbracket $, ${\left \langle 2\alpha \cdot y \right \rangle}_1$), where $x,y \in [-2^l, 2^l]$ (e.g., $l=32$) and ${\left \langle 2\alpha \cdot y \right \rangle}_1-{\left \langle 2\alpha \cdot y \right \rangle}_0 = 2\alpha \cdot y \mod N$, $S_0$ and $S_1$ correctly output $z_0$ and $z_1$ through \texttt{SMUL\textsubscript{HSS}}, respectively, such that $z_1 - z_0 = 2\alpha \cdot x \cdot y \mod N $.
\end{theorem}
\begin{proof}
In \texttt{SMUL\textsubscript{HSS}}, $S_0$ and $S_1$ compute $g_0 \leftarrow {\llbracket x \rrbracket}^{{\left \langle 2\alpha \cdot y \right \rangle}_0 }$ and $g_1 \leftarrow {\llbracket x \rrbracket}^{{\left \langle 2\alpha \cdot y \right \rangle}_1 }$, respectively.
Since ${\left \langle 2\alpha \cdot y \right \rangle}_1-{\left \langle 2\alpha \cdot y \right \rangle}_0 = 2\alpha \cdot y \mod N$, $g_1 = g_0\cdot {\llbracket x \rrbracket}^{2\alpha \cdot y} \mod N^2$ holds.
In FastPai \cite{ma2021optimized}, given a private key $\alpha$, a ciphertext $c$ and its corresponding plaintext $m$, $c^{2\alpha} \!\! \mod N^2 = (1+N)^{2\alpha m} $ holds. 
Thus, we have
\begin{align}
    g_1 = g_0\cdot {(1+N)}^{2\alpha\cdot x \cdot y} \!\!\! \mod N^2.
\end{align}

Since \texttt{DDLog}\textsubscript{N}$(g)$ in OSY21 \cite{orlandi2021rise} can correctly convert divisive shares into subtractive shares, it is easy to see that $S_0$ and $S_1$ can correctly compute $z_0 \leftarrow \texttt{DDLog\textsubscript{N}}(g_0)$ and $z_1 \leftarrow \texttt{DDLog\textsubscript{N}}(g_1)$, respectively, such that $z_1 - z_0 = 2\alpha \cdot x \cdot y \mod N $.
\end{proof}

\begin{theorem}
    Given $\llbracket x \rrbracket$ and $\llbracket y \rrbracket$, where $x,y \in [-2^l,2^l]$ (e.g., $l=32$), if $x\geq y$, \texttt{SCMP\textsubscript{HSS}} outputs $\llbracket 0 \rrbracket$, otherwise ($x<y$), \texttt{SCMP\textsubscript{HSS}} outputs $\llbracket 1 \rrbracket$.
\end{theorem}
\begin{proof}
    In \texttt{SCMP\textsubscript{HSS}}, since $x,y \in [-2^l,2^l]$, we have $x-y+1 \in [-2^{l+1}+1,2^{l+1}+1]$ and $y-x \in [-2^{l+1},2^{l+1}]$, i.e., $x-y+1,y-x \in [-2^{l+1},2^{l+1}+1]$.
    Due to $r_1\leftarrow {\{0,1\}}^\sigma \setminus \{0\}$, $ r_2 \leq \frac{N}{2}$ and $r_1 + r_2 > \frac{N}{2}$, we have $2\alpha \cdot r_1 \cdot (x-y+1)+r_2, 2\alpha \cdot r_1 \cdot (y-x)+r_2 \in [\frac{N}{2}-2^\sigma-2^{l+1+\sigma + l(\kappa)+1},\frac{N}{2}+2^{l+1+\sigma + l(\kappa)+1}+2^{\sigma+l(\kappa)+1}-2^\sigma]$, where $l(\kappa)$ represents the bit-length of the private key $\alpha$ under $\kappa$-bit security level.
    Since $l(\kappa) \gg l$ and $\frac{N}{2} \gg 2^{l+1+\sigma + l(\kappa)+1}$, we have $0<2\alpha \cdot r_1 \cdot (x-y+1)+r_2<N$ and $0<2\alpha \cdot r_1 \cdot (y-x)+r_2<N$.

    If $0<2\alpha \cdot r_1 \cdot (x-y+1)+r_2\leq \frac{N}{2}$, i.e., $x-y+1<=0$, we have $x<y$ and \texttt{SCMP\textsubscript{HSS}} outputs $\llbracket 1 \rrbracket$. If $\frac{N}{2}<2\alpha \cdot r_1 \cdot (x-y+1)+r_2<N$, i.e., $x-y+1>=1$, we have $x\geq y$ and \texttt{SCMP\textsubscript{HSS}} outputs $\llbracket 0 \rrbracket$.

    If $0< 2\alpha \cdot r_1 \cdot (y-x)+r_2\leq \frac{N}{2}$, i.e., $y-x\leq 0$, we have $x\geq y$ and \texttt{SCMP\textsubscript{HSS}} outputs $\llbracket 0 \rrbracket$.
    If $\frac{N}{2} < 2\alpha \cdot r_1 \cdot (y-x)+r_2 < N$, i.e., $y-x\geq 1$, we have $x<y$ and \texttt{SCMP\textsubscript{HSS}} outputs $\llbracket 1 \rrbracket$.

    Thus, \texttt{SCMP\textsubscript{HSS}} compares $x$ and $y$ correctly.
\end{proof}

\begin{theorem}
Given the outputs of \texttt{SMUL\textsubscript{HSS}}, i.e., the  shares $z_0 = {\left \langle 2\alpha \cdot x \cdot y \right \rangle}_0$ and $z_1 = {\left \langle 2\alpha \cdot x \cdot y \right \rangle}_1$, the share to ciphertext protocol (\texttt{S2C}) in Section \ref{conversion protocol} correctly converts shares into ciphertexts.     
\end{theorem}
\begin{proof}
    According to Theorem \ref{theorem_SMUL}, ${\left \langle {(2\alpha)}^{-1}\cdot (2\alpha) \cdot x \cdot y  \right \rangle}_0$ and ${\left \langle {(2\alpha)}^{-1}\cdot (2\alpha) \cdot x \cdot y  \right \rangle}_1$ can be correctly obtained by computing $\texttt{DDLog\textsubscript{N}}({\llbracket {(2\alpha)}^{-1} \rrbracket}^{{\left \langle 2\alpha \cdot x \cdot y  \right \rangle}_0})$ and $\texttt{DDLog\textsubscript{N}}({\llbracket {(2\alpha)}^{-1} \rrbracket}^{{\left \langle 2\alpha \cdot x \cdot y  \right \rangle}_1})$, respectively.
    In addiction, as the additive homomorphism and scalar-multiplication homomorphism of FastPai \cite{ma2021optimized}, it is easy to obtain $\llbracket x\cdot y \rrbracket$ by computing $\llbracket {\left \langle {(2\alpha)}^{-1}\cdot (2\alpha) \cdot x \cdot y  \right \rangle}_1 \rrbracket \cdot {\llbracket {\left \langle {(2\alpha)}^{-1}\cdot (2\alpha) \cdot x \cdot y  \right \rangle}_0 \rrbracket}^{-1}$.

    Thus, \texttt{S2C} correctly converts shares into ciphertexts.
\end{proof}

\begin{theorem}
    Given the output of \texttt{SCMP\textsubscript{HSS}}, i.e., a ciphertext $\llbracket \mu \rrbracket$, the ciphertext to share protocol (\texttt{C2S}) in Section \ref{conversion protocol} correctly converts a ciphertext into shares.
\end{theorem}
\begin{proof}
    \texttt{C2S} performs the secure multiplication of $\mu$ and $1$ in essence.
    According to Theorem \ref{theorem_SMUL}, \texttt{SMUL\textsubscript{HSS}} can correctly multiply $\mu$ and $1$, and outputs one share in each server, thus \texttt{C2S} correctly converts a ciphertext into shares.
\end{proof}

\subsection{Security Analysis}
KDM security \cite{black2003encryption} makes public encryption scheme secure under the condition that an adversary obtains the encrypted values of certain function with respect to the private key.
Since MORSE involves the encryption of ${(2\alpha)}^{-1}$, i.e., $\llbracket {(2\alpha)}^{-1} \rrbracket$, it is necessary to prove the KDM security of FastPai \cite{ma2021optimized}.

Following the definition of KDM security \cite{brakerski2010circular}, we difine the KDM security model of FastPai \cite{ma2021optimized} as follows.
Inspired by OSY21 \cite{orlandi2021rise}, we only consider the case with one key pair.

\begin{myDef}\label{definition 1}
Let $FastPai = (\texttt{NGen}, \texttt{KeyGen}, \texttt{Enc}, \texttt{Dec})$ be a public key encryption scheme.
For a polynomial-time adversary $\mathcal{A}$, if its advantage in the following experiment is negligible, we say that $FastPai$ is KDM secure.
\begin{enumerate}
    \item The challenger calls \texttt{KeyGen} to generate a key pair ($sk,pk$) and randomly selects $b\leftarrow \{0,1\}$.
    Next, the challenger sends $pk$ to the adversary $\mathcal{A}$.
    \item The adversary $\mathcal{A}$ randomly selects a function $f$ from $\mathcal{F}$, where $\mathcal{F}$ represents a set of functions with respect to the private key $sk$.
    Besides, the challenger sends $c$ to the adversary $\mathcal{A}$. 
    If $b=0$, $c\leftarrow \texttt{Enc}(pk,f(sk))$, otherwise ($b=1$), $c\leftarrow \texttt{Enc}(pk,0)$.
    \item The adversary $\mathcal{A}$ outputs a guess $b' \in \{0,1\}$.
\end{enumerate}

The advantage of the adversary $\mathcal{A}$ is defined as:
\begin{align}
    &\advantage[\mathcal{A}]{\text{KDM\textsuperscript{(1)}}} \nonumber \\
    &=|\prob{b'=1|b=0}-\prob{b'=1|b=1}|.
\end{align}
\end{myDef}

\begin{theorem}
    If the adversary $\mathcal{A}$ obtains the encrypted value of $f(sk)$, where $f\in \mathcal{F}$ and $F$ represents a set of functions with respect to the private key, FastPai \cite{ma2021optimized} remains secure in terms of indistinguishability, i.e., FastPai is KDM secure.
\end{theorem}
\begin{proof}
    Let the challenger generate a key pair ($sk,pk$) and randomly select $b\leftarrow \{0,1\}$.
    Besides, the adversary $\mathcal{A}$ selects a function $f$ from $\mathcal{F}$, where $\mathcal{F}$ represents a set of functions with respect to the private key $sk$.
    Next, the challenger sends $c$ to the adversary $\mathcal{A}$. 
    If $b=0$, $c\leftarrow \texttt{Enc}(pk,f(sk))$, otherwise ($b=1$), $c\leftarrow \texttt{Enc}(pk,0)$.

    In Eq. (\ref{encryption____________________________}), it is easy to observe that $\texttt{Enc}(pk,m)$ always outputs a ciphertext $c \in {\mathbb{Z}}^{*}_{{N}^{2}}$ no matter what value $m$ takes.
    Since $\texttt{Enc}(pk,f(sk))$ and $\texttt{Enc}(pk,0)$ have the same distribution, they are computationally indistinguishable.
    In other words, $\advantage[\mathcal{A}]{\text{KDM\textsuperscript{(1)}}}$ is negligible.

    Thus, we say that FastPai \cite{ma2021optimized} is KDM secure.
\end{proof}

To prove the security of the proposed protocols, we adopt the simulation paradigm \cite{micali1987play}.
Specifically, for a protocol $\mathcal{P}$ and any adversary $\mathcal{A}$ in the real model, assuming there is a simulator $\mathcal{S}$ in the ideal model, if the view of $\mathcal{A}$ from $\mathcal{P}$'s real execution and the view of $\mathcal{S}$ from $\mathcal{P}$'s ideal execution are indistinguishable, we can say that $\mathcal{P}$ is secure.

In MORSE, $S_0$ and $S_1$ are semi-honest and non-colluding. 
Besides, $S_0$ and $S_1$ are possible to be the adversaries.
In the rest of this work, we adopt $\mathcal{A}_{S_0}$ and $\mathcal{A}_{S_1}$ to denote $S_0$ and $S_1$ as adversaries, respectively.

\begin{theorem} \label{SMUL_secure}
Given ($\llbracket x \rrbracket $, ${\left \langle 2\alpha \cdot y \right \rangle}_0$) and ($\llbracket x \rrbracket $, ${\left \langle 2\alpha \cdot y \right \rangle}_1$) held by $S_0$ and $S_1$, respectively, where $x,y \in [-2^l, 2^l]$ and ${\left \langle 2\alpha \cdot y \right \rangle}_1-{\left \langle 2\alpha \cdot y \right \rangle}_0 = 2\alpha \cdot y$, in the case of adversaries $\mathcal{A}_{S_0}$ and $\mathcal{A}_{S_1}$, the proposed \texttt{SMUL\textsubscript{HSS}} in Section \ref{SMUL} securely outputs $z_0$ and $z_1$, such that $z_1-z_0 = 2\alpha \cdot x \cdot y \mod N$.
\end{theorem}

\begin{proof}
    To simulate $S_0$ and $S_1$, we construct two independent simulators $\mathcal{S}_{S_0}$ and $\mathcal{S}_{S_1}$, respectively. 

    $\mathcal{S}_{S_0}$ simulates the view of $\mathcal{A}_{S_0}$ as follows.
    \begin{itemize}
        \item $\mathcal{S}_{S_0}$ takes $\llbracket x \rrbracket $ and ${\left \langle 2\alpha \cdot y \right \rangle}_0$ as inputs.
        Next, $\mathcal{S}_{S_0}$ randomly selects $\hat{x},\hat{y}\leftarrow {\{0,1\}}^{l}$ and $\hat{\alpha}\leftarrow {\{0,1\}}^{l(\kappa)}$, where $l(\kappa)$ represents the bit-length of the private key under $\kappa$-bit security level.
        \item 
        $\mathcal{S}_{S_0}$ calls \texttt{Enc} to encrypt $\hat{x}$ into $\llbracket \hat{x} \rrbracket$.
        Next, $\mathcal{S}_{S_0}$ splits $2\hat{\alpha} \cdot \hat{y}$ into ${\left \langle 2\hat{\alpha} \cdot \hat{y} \right \rangle}_0$ and ${\left \langle 2\hat{\alpha} \cdot \hat{y} \right \rangle}_1$, where ${\left \langle 2\hat{\alpha} \cdot \hat{y} \right \rangle}_1-{\left \langle 2\hat{\alpha} \cdot \hat{y} \right \rangle}_0 = 2\hat{\alpha} \cdot \hat{y} $.
        Besides, $\mathcal{S}_{S_0}$ computes       
        $\hat{g_0} \leftarrow {\llbracket \hat{x} \rrbracket}^{{\left \langle 2\hat{\alpha} \cdot \hat{y} \right \rangle}_0 }$ and $\hat{z_0} \leftarrow \texttt{DDLog\textsubscript{N}}(\hat{g_0})$.
        \item 
        Finally, $\mathcal{S}_{S_0}$ outputs the simulation of $\mathcal{A}_{S_0}$'s view, consisting of $\llbracket \hat{x} \rrbracket$, ${\left \langle 2\hat{\alpha} \cdot \hat{y} \right \rangle}_0$, $\hat{g_0}$ and $\hat{z_0}$.
    \end{itemize}
    
    Since $\hat{x}$, $\hat{y}$ and $\hat{\alpha}$ are randomly selected, and FastPai \cite{ma2021optimized} is semantic secure, $\mathcal{S}_{S_0}$'s view from ideal model and $\mathcal{A}_{S_0}$'s view from real model are computationally indistinguishable.

    $\mathcal{S}_{S_1}$ simulates the view of $\mathcal{A}_{S_1}$ as follows.
    \begin{itemize}
        \item $\mathcal{S}_{S_1}$ takes $\llbracket x \rrbracket $ and ${\left \langle 2\alpha \cdot y \right \rangle}_1$ as inputs.
        Then, $\mathcal{S}_{S_1}$ randomly selects $\Bar{x},\Bar{y}\leftarrow {\{0,1\}}^{l}$ and $\Bar{\alpha}\leftarrow {\{0,1\}}^{l(\kappa)}$, where $l(\kappa)$ represents the bit-length of the private key under $\kappa$-bit security level.
        \item 
         $\mathcal{S}_{S_1}$ calls \texttt{Enc} to encrypt $\Bar{x}$ into $\llbracket \Bar{x} \rrbracket$.
        Then, $\mathcal{S}_{S_1}$ splits $2\Bar{\alpha} \cdot \Bar{y}$ into ${\left \langle 2\Bar{\alpha} \cdot \Bar{y} \right \rangle}_0$ and ${\left \langle 2\Bar{\alpha} \cdot \Bar{y} \right \rangle}_1$, where ${\left \langle 2\Bar{\alpha} \cdot \Bar{y} \right \rangle}_1-{\left \langle 2\Bar{\alpha} \cdot \Bar{y} \right \rangle}_0 = 2\Bar{\alpha} \cdot \Bar{y} $.
        Next, $\mathcal{S}_{S_1}$ computes       
        $\Bar{g_1} \leftarrow {\llbracket \Bar{x} \rrbracket}^{{\left \langle 2\Bar{\alpha} \cdot \Bar{y} \right \rangle}_1 }$ and $\Bar{z_1} \leftarrow \texttt{DDLog\textsubscript{N}}(\Bar{g_1})$.
        \item 
        Finally, $\mathcal{S}_{S_1}$ outputs the simulation of $\mathcal{A}_{S_1}$'s view, consisting of $\llbracket \Bar{x} \rrbracket$, ${\left \langle 2\Bar{\alpha} \cdot \Bar{y} \right \rangle}_1$, $\Bar{g_1}$ and $\Bar{z_1}$.
    \end{itemize}
    
    Since $\Bar{x}$, $\Bar{y}$ and $\Bar{\alpha}$ are randomly selected, and FastPai \cite{ma2021optimized} is semantic secure, $\mathcal{S}_{S_1}$'s view from ideal model and $\mathcal{A}_{S_1}$'s view from real model are computationally indistinguishable.

    Thus, we say that \texttt{SMUL\textsubscript{HSS}} securely outputs $z_0$ and $z_1$, such that $z_1-z_0 = 2\alpha \cdot x \cdot y \mod N$.
\end{proof}

\begin{theorem} \label{d=1/2}
Given $x,y \in [-2^l,2^l]$, $d=2\alpha \cdot r_1 \cdot (x-y+1)+r_2$ or $d=2\alpha \cdot r_1 \cdot (y-x)+r_2$, where $r_1\leftarrow {\{0,1\}}^\sigma \setminus \{0\} $, $ r_2 \leq \frac{N}{2}$ and $r_1 + r_2 > \frac{N}{2}$, if $\prob{d=2\alpha \cdot r_1 \cdot (x-y+1)+r_2}=\prob{d=2\alpha \cdot r_1 \cdot (y-x)+r_2}=\frac{1}{2}$, $S_1$ has probability of $\frac{1}{2}$ to successfully compare the relative size of $x$ and $y$.
In other words, $\prob{d>\frac{N}{2}}=\prob{d\leq \frac{N}{2}}=\frac{1}{2}$.
\end{theorem}
\begin{proof}
    We firstly consider the case $x\geq y$.
    If $\pi = 0$ ($\prob{\pi=0}=\frac{1}{2}$), we have $d=2\alpha \cdot r_1 \cdot (x-y+1)+r_2 > \frac{N}{2}$, otherwise ($\pi=1$, $\prob{\pi=1}=\frac{1}{2}$), we have $d=2\alpha \cdot r_1 \cdot (y-x)+r_2 \leq \frac{N}{2}$.
    Thus, it is easy to see that $\prob{d>\frac{N}{2}}=\prob{d\leq \frac{N}{2}}=\frac{1}{2}$
    when $x\geq y$.

    Similarly, we consider the case $x<y$.
    If $\pi = 0$ ($\prob{\pi=0}=\frac{1}{2}$), we have $d=2\alpha \cdot r_1 \cdot (x-y+1)+r_2 \leq \frac{N}{2}$, otherwise ($\pi=1$, $\prob{\pi=1}=\frac{1}{2}$), we have $d=2\alpha \cdot r_1 \cdot (y-x)+r_2 > \frac{N}{2}$.
    Thus, it is obviously that $\prob{d>\frac{N}{2}}=\prob{d\leq \frac{N}{2}}=\frac{1}{2}$
    when $x< y$.

    Taking together, we say that for any $x,y \in [-2^l,2^l]$, $\prob{d>\frac{N}{2}}=\prob{d\leq \frac{N}{2}}=\frac{1}{2}$ holds.
\end{proof}

\begin{theorem}
Given $\llbracket x \rrbracket $ and $\llbracket y \rrbracket $, where $x,y \in [-2^l, 2^l]$, in the case of adversaries $\mathcal{A}_{S_0}$ and $\mathcal{A}_{S_1}$, the proposed \texttt{SCMP\textsubscript{HSS}} in Section \ref{SCMP} securely compares $x$ and $y$.     
\end{theorem}

\begin{proof}
    To simulate $S_0$ and $S_1$, we construct two independent simulators $\mathcal{S}_{S_0}$ and $\mathcal{S}_{S_1}$, respectively. 

    $\mathcal{S}_{S_0}$ simulates the view of $\mathcal{A}_{S_0}$ as follows.
    \begin{itemize}
        \item $\mathcal{S}_{S_0}$ takes $\llbracket x \rrbracket$, $\llbracket y \rrbracket$ and $\llbracket \mu_0 \rrbracket$ as inputs.
        Next, $\mathcal{S}_{S_0}$ randomly selects $\hat{x},\hat{y} \in [-2^l, 2^l]$, $\hat{r_1}\leftarrow {\{0,1\}}^\sigma \setminus \{0\}$, $\hat{r_2}\in (\frac{N}{2}-\hat{r_1},\frac{N}{2}]$, $\hat{\pi}\leftarrow \{0,1\}$ and $\hat{\alpha}\leftarrow {\{0,1\}}^{l(\kappa)}$, where $l(\kappa)$ represents the bit-length of the private key under $\kappa$-bit security level.
        \item 
        $\mathcal{S}_{S_0}$ calls $\texttt{Enc}$ to encrypt $\hat{x}$, $\hat{y}$, 0 and 1 into $\llbracket \hat{x} \rrbracket$, $\llbracket \hat{y} \rrbracket$, $\llbracket 0 \rrbracket$ and $\llbracket 1 \rrbracket$, respectively. 
        Next, $\mathcal{S}_{S_0}$ splits $2\hat{\alpha}$ into ${\left \langle 2\hat{\alpha} \right \rangle}_0$ and ${\left \langle 2\hat{\alpha} \right \rangle}_1$, such that ${\left \langle 2\hat{\alpha} \right \rangle}_1-{\left \langle 2\hat{\alpha} \right \rangle}_0 = 2\hat{\alpha}$.
        \item 
        If $\hat{\pi}=0$, $\mathcal{S}_{S_0}$ computes $\hat{D} \leftarrow {(\llbracket \hat{x} \rrbracket \cdot {\llbracket \hat{y} \rrbracket}^{-1}\cdot \llbracket 1 \rrbracket)}^{\hat{r_1}}$, otherwise ($\hat{\pi}=1$), $\mathcal{S}_{S_0}$ computes $\hat{D} \leftarrow {(\llbracket \hat{y} \rrbracket \cdot {\llbracket \hat{x} \rrbracket}^{-1})}^{\hat{r_1}}$.
        Besides, $\mathcal{S}_{S_0}$ computes $\hat{z_0} \leftarrow \texttt{DDLog\textsubscript{N}}(\hat{D}^{{\left \langle 2\hat{\alpha} \right \rangle}_0})$ and $\hat{z_0} \leftarrow \hat{z_0}-\hat{r_2}$.
        Next, if $\Hat{\pi} = 0$, $\mathcal{S}_{S_0}$ computes $\llbracket \hat{\mu} \rrbracket  \leftarrow \llbracket {\mu}_0 \rrbracket $, otherwise ($\hat{\pi} = 1$), $\mathcal{S}_{S_0}$ computes $\llbracket \hat{\mu} \rrbracket  \leftarrow \llbracket 1 \rrbracket \cdot {\llbracket {\mu}_0\rrbracket }^{-1}$.
        Moreover, $\mathcal{S}_{S_0}$ computes $\llbracket \hat{\mu} \rrbracket  \leftarrow \llbracket \hat{\mu} \rrbracket \cdot \llbracket 0 \rrbracket$.
        \item 
        Finally, $\mathcal{S}_{S_0}$ outputs the simulation of $\mathcal{A}_{S_0}$'s view, consisting of $\llbracket \hat{x} \rrbracket$, $\llbracket \hat{y} \rrbracket$, $\hat{D}$, $\llbracket \hat{\mu} \rrbracket$ and $\hat{z_0}$.
    \end{itemize}

    Since $\hat{x}$, $\hat{y}$, $\hat{r_1}$, $\hat{r_2}$ and $\hat{\alpha}$ are randomly selected, and FastPai \cite{ma2021optimized} is semantic secure, $\mathcal{S}_{S_0}$'s view from ideal model and $\mathcal{A}_{S_0}$'s view from real model are computationally indistinguishable.

    $\mathcal{S}_{S_1}$ simulates the view of $\mathcal{A}_{S_1}$ as follows.
    \begin{itemize}
    \item $\mathcal{S}_{S_1}$ takes $D$ and $z_0$ as inputs.
    Next, $\mathcal{S}_{S_1}$ randomly selects $\Bar{x},\Bar{y} \in [-2^l, 2^l]$, $\Bar{r_1}\leftarrow {\{0,1\}}^\sigma \setminus \{0\}$, $\Bar{r_2}\in (\frac{N}{2}-\Bar{r_1},\frac{N}{2}]$, $\Bar{\pi}\leftarrow \{0,1\}$ and $\Bar{\alpha}\leftarrow {\{0,1\}}^{l(\kappa)}$, where $l(\kappa)$ represents the bit-length of the private key under $\kappa$-bit security level.
    \item 
    $\mathcal{S}_{S_1}$ splits $2\Bar{\alpha}$ into ${\left \langle 2\Bar{\alpha} \right \rangle}_0$ and ${\left \langle 2\Bar{\alpha} \right \rangle}_1$, such that ${\left \langle 2\Bar{\alpha} \right \rangle}_1-{\left \langle 2\Bar{\alpha} \right \rangle}_0 = 2\Bar{\alpha}$.
    Besides, if $\Bar{\pi}=0$, $\mathcal{S}_{S_1}$ calls $\texttt{Enc}$ to encrypt $2\Bar{\alpha}\cdot \Bar{r_1} \cdot (\Bar{x}-\Bar{y}+1)+\Bar{r_2}$ into $\Bar{D}$, otherwise ($\Bar{\pi}=1$), $\mathcal{S}_{S_1}$ calls $\texttt{Enc}$ to encrypt $2\Bar{\alpha}\cdot \Bar{r_1} \cdot (\Bar{y}-\Bar{x})+\Bar{r_2}$ into $\Bar{D}$.
    \item 
    $\mathcal{S}_{S_1}$ computes $\Bar{z_0} \leftarrow \texttt{DDLog\textsubscript{N}}(D^{{\left \langle 2\Bar{\alpha} \right \rangle}_0})$ and $\Bar{z_0} \leftarrow \Bar{z_0}-\Bar{r_2}$.
    Besides, $\mathcal{S}_{S_1}$ computes $\Bar{z_1} \leftarrow \texttt{DDLog\textsubscript{N}}(D^{{\left \langle 2\Bar{\alpha} \right \rangle}_1})$ and $\Bar{d} \leftarrow \Bar{z_1}-\Bar{z_0} \mod N$.
    \item 
    Next, if $\Bar{x}\geq \Bar{y}$, $\mathcal{S}_{S_1}$ sets $\Bar{\mu_0}=0$, otherwise ($\Bar{x} < \Bar{y}$), $\mathcal{S}_{S_1}$ sets $\Bar{\mu_0}=1$.
    Then, $\mathcal{S}_{S_1}$ calls $\texttt{Enc}$ to encrypt $\Bar{\mu_0}$ into $\llbracket \Bar{\mu_0} \rrbracket$.
    \item 
    If $\Bar{\pi} = 0$, $\mathcal{S}_{S_1}$ outputs the simulation of $\mathcal{A}_{S_1}$'s view, consisting of $\Bar{D}$, $\Bar{z_0}$, $\Bar{z_1}$, $2\Bar{\alpha}\cdot \Bar{r_1} \cdot (\Bar{x}-\Bar{y}+1)+\Bar{r_2}$ and $\llbracket \Bar{\mu_0} \rrbracket$, otherwise ($\Bar{\pi}$=1), $\mathcal{S}_{S_1}$ outputs the simulation of $\mathcal{A}_{S_1}$'s view, consisting of $\Bar{D}$, $\Bar{z_0}$, $\Bar{z_1}$, $2\Bar{\alpha}\cdot \Bar{r_1} \cdot (\Bar{y}-\Bar{x})+\Bar{r_2}$ and $\llbracket \Bar{\mu_0} \rrbracket$.
    \end{itemize}

    Since the one time key encryption scheme $x+r$ is proven to be chosen-plaintext attack secure by SOCI \cite{zhao2022soci}, $2\Bar{\alpha}\cdot \Bar{r_1} \cdot (\Bar{x}-\Bar{y}+1)+\Bar{r_2}$ and $2\alpha \cdot r_1 \cdot (x-y+1)+r_2$ are computationally indistinguishable. 
    Similarly, $2\Bar{\alpha}\cdot \Bar{r_1} \cdot (\Bar{y}-\Bar{x})+\Bar{r_2}$ and $2\alpha \cdot r_1 \cdot (y-x)+r_2$ are also computationally indistinguishable. 
    In addition, according to Theorem \ref{d=1/2}, we have $\prob{d>\frac{N}{2}}=\prob{d\leq \frac{N}{2}}=\frac{1}{2}$.
    Besides, since $\Bar{x}$, $\Bar{y}$ and $\Bar{\alpha}$ are randomly selected, and FastPai \cite{ma2021optimized} is semantic secure, $\mathcal{S}_{S_1}$'s view from ideal model and $\mathcal{A}_{S_1}$'s view from real model are computationally indistinguishable.

    Thus, we say that \texttt{SCMP\textsubscript{HSS}} can compare $x$ and $y$ securely.
\end{proof}

\begin{theorem}
    Given the outputs of \texttt{SMUL\textsubscript{HSS}}, i.e., the shares $z_0 = {\left \langle 2\alpha \cdot x \cdot y \right \rangle}_0$ and $z_1 = {\left \langle 2\alpha \cdot x \cdot y \right \rangle}_1$, in the case of adversaries $\mathcal{A}_{S_0}$ and $\mathcal{A}_{S_1}$, the share to ciphertext protocol (\texttt{S2C}) in Section \ref{conversion protocol} securely converts shares into ciphertexts.  
\end{theorem}
\begin{proof}
    The share to ciphertext protocol (\texttt{S2C}) is based on \texttt{SMUL\textsubscript{HSS}}.
    According to Theorem \ref{SMUL_secure}, \texttt{SMUL\textsubscript{HSS}} is secure to achieve multiplication, so \texttt{S2C} is also secure to convert shares into ciphertexts. 
\end{proof}

\begin{theorem}
    Given the output of \texttt{SCMP\textsubscript{HSS}}, i.e., a ciphertext $\llbracket \mu \rrbracket$, in the case of adversaries $\mathcal{A}_{S_0}$ and $\mathcal{A}_{S_1}$, the ciphertext to share protocol (\texttt{C2S}) in Section \ref{conversion protocol} securely converts a ciphertext into shares.
\end{theorem}

\begin{proof}
    The ciphertext to share protocol (\texttt{C2S}) is based on \texttt{SMUL\textsubscript{HSS}}.
    According to Theorem \ref{SMUL_secure}, \texttt{SMUL\textsubscript{HSS}} is secure to achieve multiplication, so \texttt{C2S} is also secure to convert a ciphertext into shares.
\end{proof}

\begin{table*}[!ht]
\centering
\caption{Runtime comparison between MORSE and related schemes (Unit: \lowercasenote{ms})}
\begin{tabular}{@{}ccccccc@{}}
\toprule
Algorithms            & MORSE          & FGJS17 \cite{fazio2017homomorphic} & OSY21 \cite{orlandi2021rise} & RS21 \cite{roy2021large} & BFV \cite{fan2012somewhat}           & CKKS \cite{cheon2017homomorphic}\\ \midrule
Secure multiplication & \textbf{6.3}  & 3978.1                             & 58.4                         & 92.5                     & 8.8                                  & 9.6  \\
Secure addition       & \textbf{0.01} & 0.01                               & 0.01                         & 0.02                     & 0.05                                 & 0.11  \\
Secure subtraction    & 0.06          & 0.06                               & 0.06                         & 0.11                     & \textbf{0.04}                        & 0.08 \\
Scalar-multiplication & \textbf{0.08} & 0.08                               & 0.08                         & 0.13                     & 1.54                                 & 3.17 \\ \bottomrule
\end{tabular}
\label{Comparison of between the proposed HSS scheme and related schemes}
\end{table*}

\section{Experimental Evaluation}\label{Section_7}
\subsection{Experimental Settings}
We implement MORSE on a single host (CPU: 11th Gen Intel(R) Core(TM) i7-11700 @ 2.50GHz; Memory: 16.0 GB) with Python 3.6.13 and gmpy2-2.1.5, and we treat each party as a separate process to simulate LAN network.
Note that the experiments are conducted using a single thread, with each experiment repeats 500 times, and the average of these values taken as the results.
In the rest of this work, we represent the bit-length of $N$ as $|N|$.
In our experimental settings, we set the security level as 128-bit security level, i.e., $|N|=3072$.

\subsection{Comparison between MORSE and Existing Solutions}
In this subsection, we compare the basic cryptographic operations of MORSE with those of the related HSS schemes (FGJS17 \cite{fazio2017homomorphic}, OSY21 \cite{orlandi2021rise} and RS21 \cite{roy2021large}), as well as the famous fully homomorphic encryption schemes (BFV \cite{fan2012somewhat} and CKKS \cite{cheon2017homomorphic}).
Note that BFV and CKKS are implemented in TenSEAL \cite{benaissa2021tenseal,Tenseal_code}, a framework that is constructed on Microsoft SEAL \cite{Microsoft_SEAL} and provides Python API.
Besides, the parameters for $\texttt{DDLog}$ of FGJS17 \cite{fazio2017homomorphic} are configured as $\delta=0.05$ and $M=10001$, and the parameter $s$ is set as 2 for the Damg{\aa}rd–Jurik public-key encryption scheme adopted in RS21 \cite{roy2021large} (i.e., the modulus of ciphertexts is $N^3$).
To determine the exponent $b$ of $xy=g^b$ by $\texttt{DDLog}$, where $x$ and $y$ are held by two different parties, FGJS17 \cite{fazio2017homomorphic} claims that $b$ should be a small number and $b<M$.
Therefore, we randomly select the test data from $[1,100]$ in our experimental settings, ensuring the maximum of $b$ is 10000, which complies with the condition $b<M$.

As depicted in Table \ref{Comparison of between the proposed HSS scheme and related schemes}, MORSE demonstrates superior performance in secure multiplication, secure addition and scalar-multiplication.
Among MORSE and the related schemes (FGJS17 \cite{fazio2017homomorphic}, OSY21 \cite{orlandi2021rise} and RS21 \cite{roy2021large}), the secure multiplication is conducted in parallel across the two servers, with the time complexities being nearly identical on the two servers.
Hence, we consider the longest runtime of the two servers as the benchmark result.
Specifically, the proposed secure multiplication (\texttt{SMUL\textsubscript{HSS}}) offers approximately 9.3 times greater efficiency compared to the state-of-the-art HSS scheme OSY21 \cite{orlandi2021rise}.
Furthermore, \texttt{SMUL\textsubscript{HSS}} outperforms the famous FHE schemes BFV \cite{fan2012somewhat} and CKKS \cite{cheon2017homomorphic}.
Despite BFV \cite{fan2012somewhat} has slight advantage in secure subtraction, MORSE markedly outperforms BFV \cite{fan2012somewhat} in secure multiplication, secure addition and scalar-multiplication.
Note that the runtime for secure multiplication in FGJS17 \cite{fazio2017homomorphic} significantly exceeds that of other HSS schemes.
One reason is that the \texttt{DDLog} in FGJS17 \cite{fazio2017homomorphic} necessitates iterative computations to resolve the distributed discrete logarithm, whereas other HSS schemes' \texttt{DDLog} can achieve this directly.
In short, MORSE features efficient \texttt{DDLog} and reduced number of modular exponential operations, leading to a great improvement in efficiency.
Arguably, MORSE enjoys superior performance in linear operations.

To enable the outputs of secure multiplication to serve as the inputs for subsequent secure multiplication, FGJS17 \cite{fazio2017homomorphic} and OSY21 \cite{orlandi2021rise} split the private key into multiple shares and involve complex operations for making the outputs ${\left \langle  x \cdot y \right \rangle}_0$ and ${\left \langle  x \cdot y \right \rangle}_1$ be ${\left \langle sk \cdot x \cdot y \right \rangle}_0$ and ${\left \langle sk \cdot x \cdot y \right \rangle}_1$, respectively, where $sk$ is the private key of FGJS17 \cite{fazio2017homomorphic} and OSY21 \cite{orlandi2021rise}.
Note that Table \ref{Comparison of between the proposed HSS scheme and related schemes} omits these operations and only considers single multiplication for simplicity.
To compare the performance of performing consecutive operations between MORSE and existing HSS schemes, we evaluate a polynomial $f(x)=x^5+x^4+x^3+x^2+x+1$ (i.e, 4 consecutive secure multiplications).
Since FGJS17 \cite{fazio2017homomorphic} is theoretically more time-consuming than OSY21 \cite{orlandi2021rise} and RS21 \cite{roy2021large} in performing consecutive secure multiplication, we only compare MORSE with OSY21 \cite{orlandi2021rise} and RS21 \cite{roy2021large}.
With the modulus reduction technology introduced in Section \ref{modulus_reduction}, MORSE is able to transfer the modulus of shares from $N$ to a smaller modulus to make consecutive secure multiplications more efficient. 
However, this method is not applicable for OSY21 \cite{orlandi2021rise} and RS21 \cite{roy2021large}.
As presented in Table \ref{Comparison of performing consecutive operations}, MORSE significantly outperforms the existing schemes in terms of runtime.
To enable consecutive operations, OSY21 \cite{orlandi2021rise} and RS21 \cite{roy2021large} convert the shares with a Paillier modulus into the shares over $\mathbb{Z}$, incurring a negligible calculation error.
Fortunately, MORSE discards this conversion and enjoys no calculation error.
Arguably, MORSE shows a superior advantage in performing consecutive operations.

\begin{table}[]
\caption{Runtime comparison of performing consecutive operations}
\label{Comparison of performing consecutive operations}
\centering
\begin{tabular}{@{}llll@{}}
\toprule
Metrics       & MORSE                    & OSY21 \cite{orlandi2021rise} & RS21 \cite{roy2021large}  \\ \midrule
Runtime (Unit: ms) & \multicolumn{1}{c}{24.2} & \multicolumn{1}{c}{1274.8}   & \multicolumn{1}{c}{363.2} \\ 
Calculation error  & \multicolumn{1}{c}{no error}                 & \multicolumn{1}{c}{negligible}                   & \multicolumn{1}{c}{negligible} \\
\bottomrule
\end{tabular}
\end{table}

\begin{table*}[]
\caption{Comparison of different secure comparison protocols}
\centering
\begin{tabular}{@{}cccccc@{}}
\toprule
Metrics                       & MORSE           & SOCI\textsuperscript{+} \cite{zhao2024soci+} & SOCI \cite{zhao2022soci} & POCF \cite{liu2016privacy} & DGK \cite{veugen2018correction} \\ \midrule
Runtime (Unit: ms)            & \textbf{13.9}  & 47.7                                         & 141.3                    & 144.0                      & 163.5                           \\
Communication cost (Unit: KB) & \textbf{1.874} & 2.247                                        & 2.249                    & 2.249                      & 27.740                          \\ \bottomrule
\end{tabular}
\label{Comparison of different secure comparison protocols}
\end{table*}

To evaluate the performance of the proposed secure comparison protocol (\texttt{SCMP\textsubscript{HSS}}), we compare it with SOCI\textsuperscript{+} \cite{zhao2024soci+}, SOCI \cite{zhao2022soci}, POCF \cite{liu2016privacy} and a famous comparison protocol DGK \cite{veugen2018correction}.
In our experimental settings of secure comparison, we set the data length as 32 (i.e., $l=32$).
As illustrated in Table \ref{Comparison of different secure comparison protocols}, the proposed \texttt{SCMP\textsubscript{HSS}} demonstrates significantly superior performance over other schemes in terms of runtime and communication costs.
Notably, compared with the state-of-the-art SOCI\textsuperscript{+} \cite{zhao2024soci+}, the proposed \texttt{SCMP\textsubscript{HSS}} offers a runtime improvement of up to 3.4 times and achieves a reduction of $16.6\%$ in communication costs.

\begin{figure}[h]
    \centering
    \subfloat[The runtime under a varying $N$]{
        \includegraphics[width=0.23\textwidth]{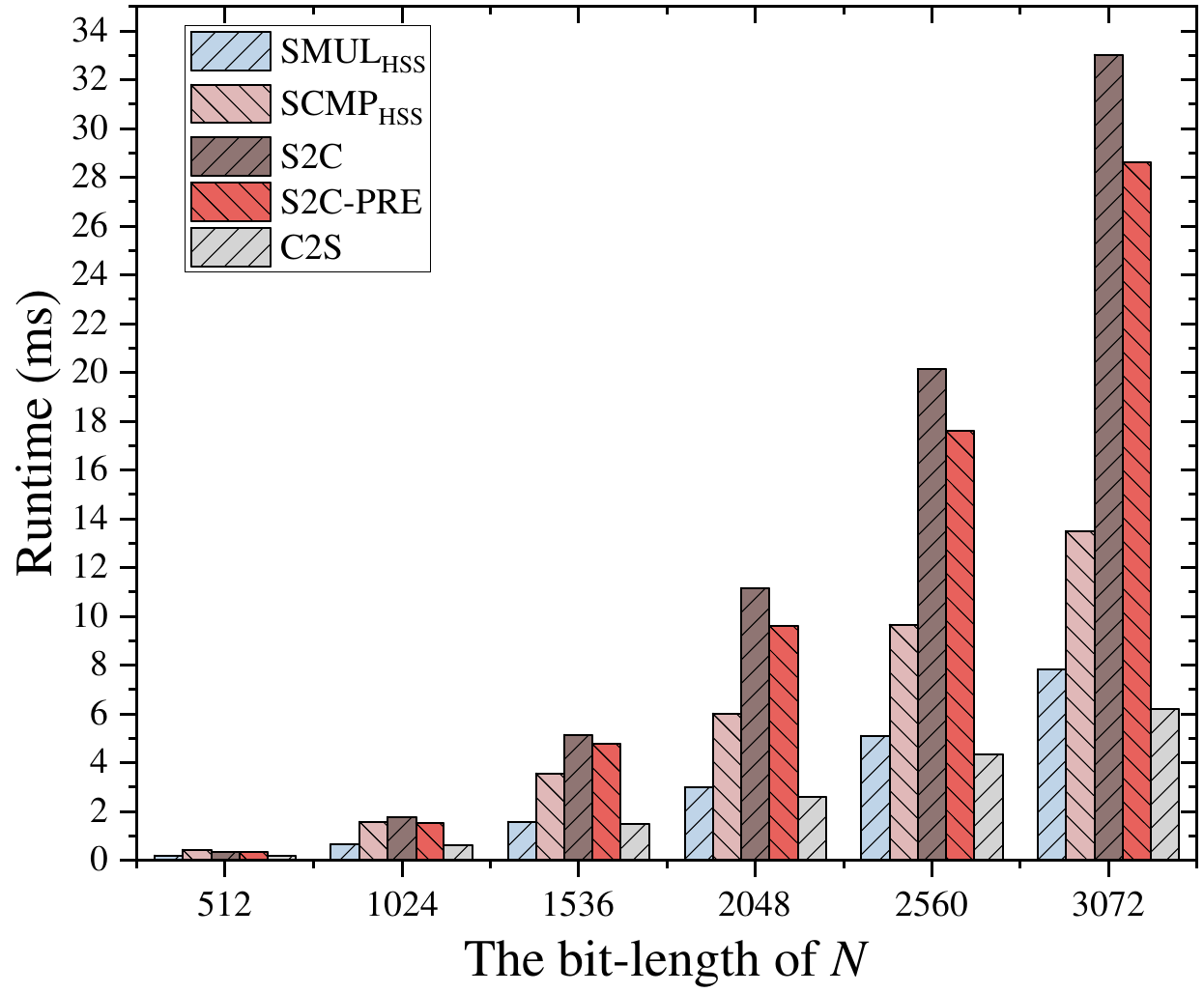}
        \label{The runtime under a varying N}
    }
    \hfill
    \subfloat[The communication cost under a varying $N$]{
        \includegraphics[width=0.23\textwidth]{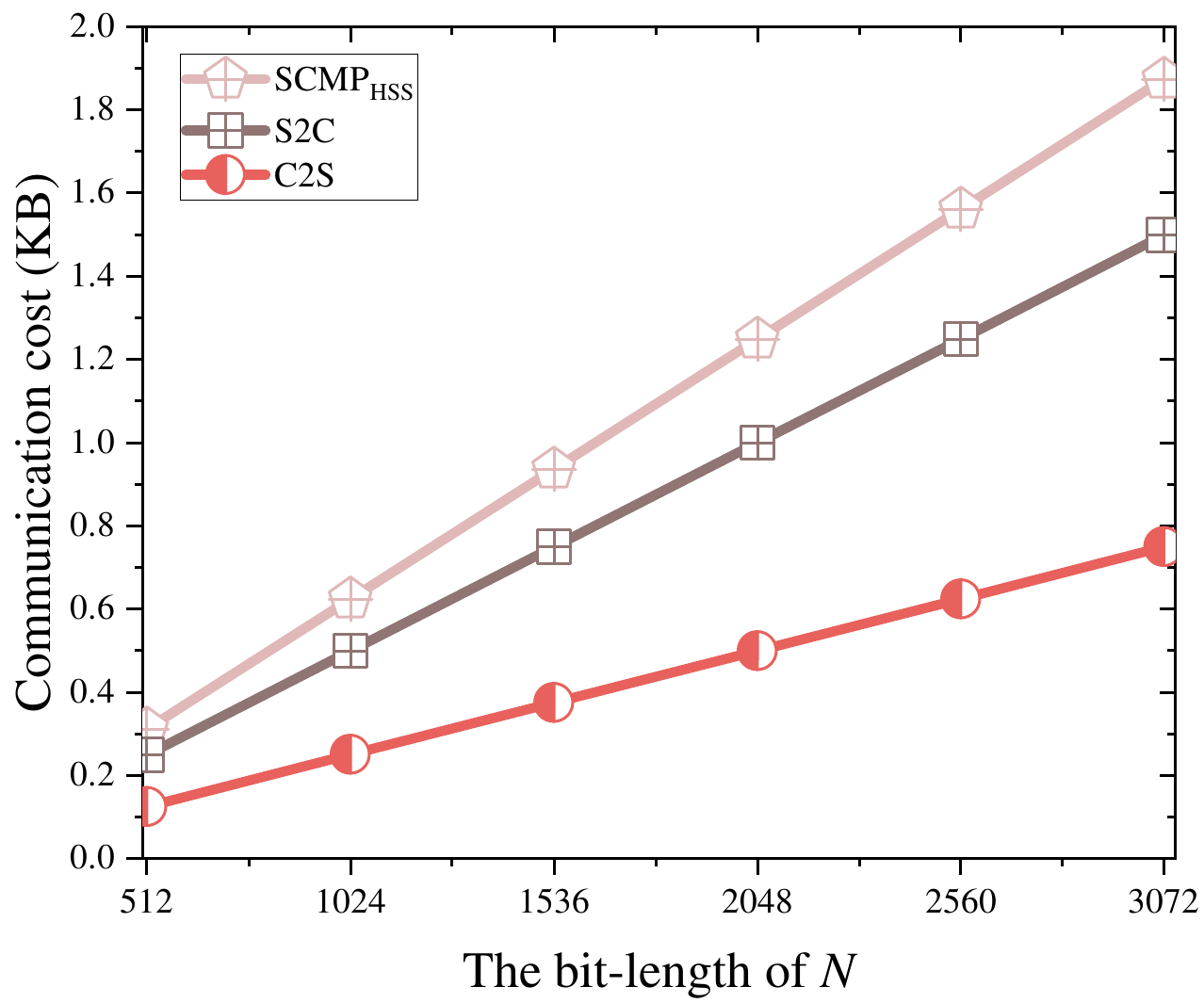}
        \label{The communication cost under a varying N}
    }
    \caption{The performance of MORSE under a varying $N$ ($l=32$)}
    \label{The performance of HSS under a varying N}
\end{figure}

\begin{figure}[h]
    \centering
    \subfloat[The runtime under a varying $l$]{
        \includegraphics[width=0.23\textwidth]{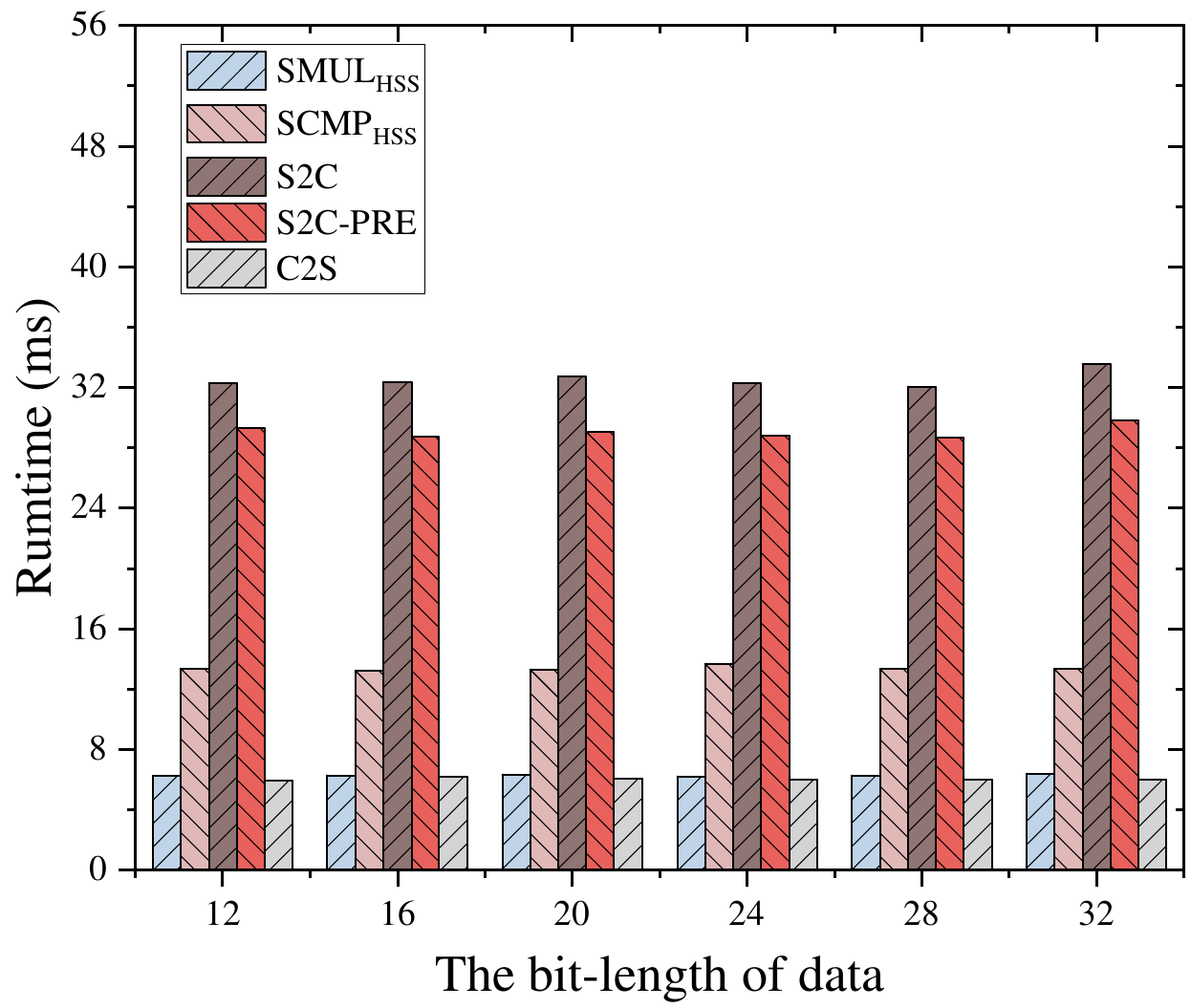}
        \label{The runtime under a varying L}
    }
    \hfill
    \subfloat[The communication cost under a varying $l$]{
        \includegraphics[width=0.23\textwidth]{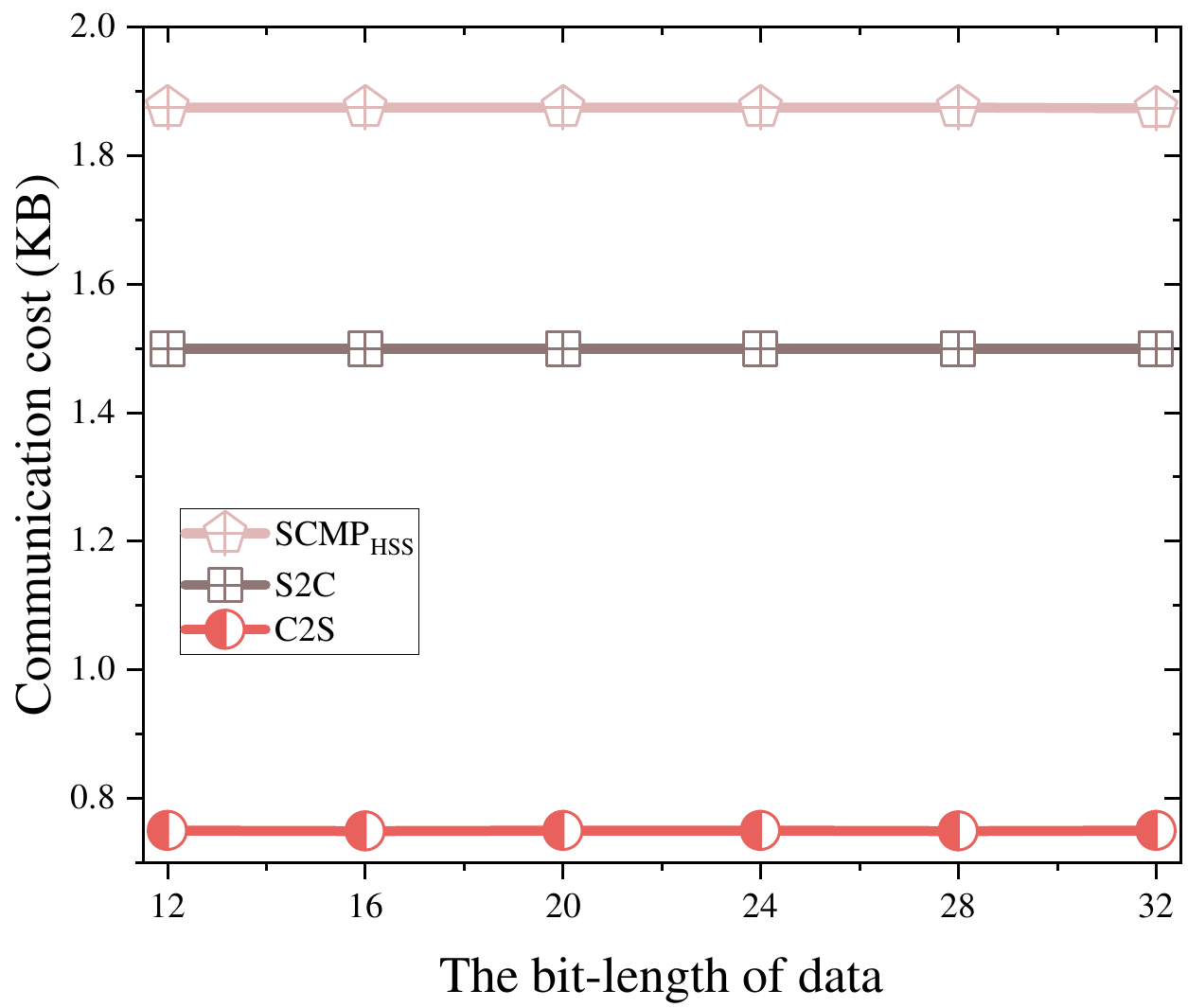}
        \label{The communication cost under a varying L}
    }
    \caption{The performance of MORSE under a varying $l$ ($|N|=3072$)}
    \label{The performance of HSS under a varying L}
\end{figure}

\subsection{Performance Analysis of MORSE}
Given one pair of base and exponent $(b,e)$, it requires $1.5|e|$ multiplication operations to compute $b^e$ \cite{knuth2014art}, where $|e|$ is the bit-length of $e$.
Considering that the exponentiation operation consumes more resources than addition operation and multiplication operation \cite{liu2016efficient_multiple_keys}, we only consider the time-consuming modulo exponentiation operation in our analysis.
The data length and the modulo multiplication operation are denoted as $l$ and \texttt{mul}, respectively.

\texttt{SMUL\textsubscript{HSS}} costs $1.5\cdot (|2\alpha|+ (l+\kappa))$ \texttt{mul} for both $S_0$ and $S_1$, and these operations are executed in parallel on the two servers.
\texttt{SCMP\textsubscript{HSS}} costs $1.5\cdot(\sigma + |2\alpha|+\kappa)$ \texttt{mul} for $S_0$ and $1.5\cdot (|2\alpha|+ \kappa)$ \texttt{mul} for $S_1$.
\texttt{S2C} costs $1.5\cdot (|N|+4\kappa)$ \texttt{mul} for both $S_0$ and $S_1$, and these operations are executed in parallel on the two servers.
By adopting the pre-computation table for speeding up \texttt{Enc} in the scheme \cite{ma2021optimized}, the costs of $\texttt{S2C}$ can be reduced to $1.5\cdot |N|+\lceil 
\frac{4\kappa}{b} \rceil $ (i.e., $b=5$) \texttt{mul}.
\texttt{C2S} costs $1.5\cdot(|2\alpha|+\kappa)$ \texttt{mul} for both $S_0$ and $S_1$, and these operations are executed in parallel on the two servers.
Since $\kappa$ and $\alpha$ are related to $N$, we can see that the runtime of the proposed protocols mainly depends on $N$.
As shown in Figs. \ref{The performance of HSS under a varying N}(a) and \ref{The performance of HSS under a varying L}(a), the runtime of the proposed protocols increases with $|N|$ increasing, and the runtime is independent of $l$.

For communication costs, the communication costs of \texttt{SCMP\textsubscript{HSS}}, \texttt{S2C} and \texttt{C2S} are $5|N|$ bits, $4|N|$ bits and $2|N|$ bits, respectively.
As shown in Figs. \ref{The performance of HSS under a varying N}(b) and \ref{The performance of HSS under a varying L}(b), the communication costs of the proposed protocols increase with $|N|$ increasing, and the communication costs are independent of $l$.

\section{Conclusion} \label{Section_8}
In this work, we proposed MORSE, an efficient homomorphic secret sharing based on FastPai \cite{ma2021optimized}, which addresses the limitations in existing schemes.
Specifically, we presented a secure multiplication protocol \texttt{SMUL\textsubscript{HSS}} enabling two servers to perform secure multiplication on Paillier ciphertexts in local.
Besides, built top on \texttt{SMUL\textsubscript{HSS}}, we proposed an secure comparison protocol \texttt{SCMP\textsubscript{HSS}} to enable non-linear operation.
Furthermore, we designed two protocols to enable the conversion between shares and ciphertexts, i.e., a share to ciphertext protocol \texttt{S2C} and a ciphertext to share protocol \texttt{C2S}.
We provided rigorous analyses depicting that the proposed protocols correctly and securely output the results.
Experimental results demonstrated that MORSE significantly outperforms the related solutions.
In future work, we intend to extend MORSE to support more types of secure computation and operations on floating-point numbers. 


\bibliographystyle{IEEEtran}
\normalem
\bibliography{hss}

\end{document}